\documentclass[12pt, draftclsnofoot, onecolumn]{IEEEtran}

\usepackage{cite}
\usepackage[pdftex]{graphicx}
\graphicspath{{Figure/}}
\usepackage[cmex10]{amsmath}
\usepackage{amsthm}
\usepackage{amssymb}
\usepackage{algorithmic}
\usepackage{array}
\usepackage{color}
\usepackage{epstopdf}
\usepackage{amsfonts}

\usepackage{pgfplots}
\pgfplotsset{compat=1.3}
\usepackage{tikz}							
\usetikzlibrary{shapes}
\usetikzlibrary{spy}
\usetikzlibrary{circuits}
\usetikzlibrary{arrows}

\newcommand{\mat}[1]{\boldsymbol{\mathrm{#1}}}

\usepackage{mathtools}

\newtheorem{theorem}{Theorem}

\newtheorem{proposition}{Proposition}
\newtheorem{lemma}{Lemma}

\ifCLASSOPTIONcompsoc
  \usepackage[caption=false,font=normalsize,labelfont=sf,textfont=sf]{subfig}
\else
  \usepackage[caption=false,font=footnotesize]{subfig}
\fi
\usepackage{dblfloatfix}

\begin{document}

\title{{CSI-based versus RSS-based Secret-Key Generation under Correlated Eavesdropping}}

\author{Fran\c{c}ois Rottenberg,~\IEEEmembership{Member,~IEEE,}
	Trung-Hien~Nguyen,~\IEEEmembership{Member,~IEEE,}
	Jean-Michel Dricot,~\IEEEmembership{Member,~IEEE,}
	Fran\c{c}ois Horlin,~\IEEEmembership{Member,~IEEE,}\\
	and~J\'er\^ome Louveaux,~\IEEEmembership{Member,~IEEE.}
	\thanks{The research reported herein was partly funded by the Fonds national de la recherche scientifique (F.R.S.-FNRS). Part of the material in this paper has been submitted to IEEE PIMRC 2020 \cite{rottenberg2020impact}.} 
	\thanks{Fran\c{c}ois Rottenberg and J\'er\^ome Louveaux are with the Universit\'e catholique de Louvain, 1348 Louvain-la-Neuve, Belgium (e-mail: francois.rottenberg@uclouvain.be).}
	\thanks{Fran\c{c}ois Rottenberg, Trung-Hien Nguyen, Jean-Michel Dricot and Fran\c{c}ois Horlin are with the Universit\'e libre de Bruxelles, 1050 Brussel, Belgium.}
}

%



\maketitle

\begin{abstract}
	Physical-layer security (PLS) has the potential to strongly enhance the overall system security as an alternative to or in combination with conventional cryptographic primitives usually implemented at higher network layers. Secret-key generation relying on wireless channel reciprocity is an interesting solution as it can be efficiently implemented at the physical layer of emerging wireless communication networks, while providing information-theoretic security guarantees. In this paper, we investigate and compare the secret-key capacity based on the sampling of the entire complex channel state information (CSI) or only its envelope, the received signal strength (RSS). Moreover, as opposed to previous works, we take into account the fact that the eavesdropper's observations might be correlated and we consider the high signal-to-noise ratio (SNR) regime where we can find simple analytical expressions for the secret-key capacity. As already found in previous works, we find that RSS-based secret-key generation is heavily penalized as compared to CSI-based systems. At high SNR, we are able to precisely and simply quantify this penalty: a halved pre-log factor and a constant penalty of about 0.69 bit, which disappears as Eve's channel gets highly correlated.
\end{abstract}

\begin{IEEEkeywords}
Secret-Key Generation, RSS, CSI, Physical-Layer Security.
\end{IEEEkeywords}

%
\IEEEpeerreviewmaketitle

\section{Introduction}

\subsection{Problem Statement}

We consider in this paper the problem of generating secret keys between two legitimate users (Alice and Bob), subject to an illegitimate user (Eve) trying to recover the key. Maurer \cite{Maurer1993} and Ahlswede and Csisz\'{a}r \cite{Csiszar1993} were the first to analyze the problem of generating a secret key from correlated observations. In the source model (see Fig.~\ref{fig:source_model}), Alice, Bob and Eve observe the realizations of a discrete memoryless source. From their sequence of observations, Alice and Bob have to distill an identical key that remains secret from Eve. Moreover, Alice and Bob have access to a public error-free authenticated channel with unlimited capacity. 
This helps them to perform \textit{information reconciliation}, \textit{i.e.}, exchanging a few parity bits so as to agree on a common sequence of symbols. However, since the channel is public, Eve can gain information about the secret key from these parity bits, on top of her own channel observations that can also be correlated with Alice and Bob observations. This is why \textit{privacy amplification} is usually implemented after \textit{information reconciliation}, which consists in reducing the size of the key, through, \textit{e.g.}, universal hashing, so that Eve information about the key is completely eliminated. Upper and lower bounds for the secret-key capacity, defined as the number of secret bits that can be generated per observation of the source, were derived in \cite{Csiszar1993,Maurer1993}.

A practical source of common randomness at Alice and Bob consists of the wireless channel reciprocity, which implies that the propagation channel from Alice to Bob and from Bob to Alice is identical if both are measured within the same channel coherence time and at the same frequency. At successive coherence times, Alice and Bob can repeatedly sample the channel by sending each other a pilot symbol so as to obtain a set of highly correlated observations and finally start a key-distillation procedure. In this paper, we investigate the secret-key capacity relying on the entire complex channel state information (CSI) or only on the channel envelope, sometimes also referred to as received signal strength (RSS)\footnote{We focus the whole study in this paper on the envelope of the channel, not its power. However, the final results in terms of capacity are equivalent given the one-to-one relationship between envelope and power.}. We also consider the case where Eve's observations are correlated with the ones of Alice and Bob, which can occur in many practical situations, as explained in next section.

\subsection{State of the Art}

This study falls into the broad field of physical-layer security (PLS), which has attracted much interest in the recent decade as a competitive candidate to provide authentication, integrity and confidentiality in future communication networks \cite{Yang2015s5wc,Wu2018asopls,Hamamreh2019}. We refer to \cite{bloch2011physical} for an overview on the area. In the context of secret-key generation based on wireless reciprocity, there has been a large amount of related works, both from theoretical and experimental aspects \cite{Zeng2015plkg,Jorswieck2015,Zhang2016b}. Many works have considered using RSS as a source of randomness for secret-key generation, including multiple-antenna systems \cite{AzimiSadjadi2007,Jana2009,Zeng2010,Patwari2010,Liu2012,Guillaume2015,Zenger2016}. The choice of using RSS over full CSI has two main advantages: i) as opposed to CSI, RSS indicators are usually available at the higher layers of the communication layers, allowing for simple implementation of the key distillation procedure, relying on the legacy network infrastructure and ii) RSS is intrinsically more robust to phase offsets between Alice and Bob, relaxing constraints on the hardware, the synchronization and the reciprocity calibration.

The main disadvantage of RSS-based secret-key generation is that it does not use the full channel information and thus achieves a lower secret-key capacity than its CSI-based counterpart. This solution has also been studied in many different works. CSI-based secret-key capacity is generally easier to characterize analytically, which has been done in a large number of works \cite{Ye2006,Ye2010}, relying on multi-antenna systems \cite{wong2009secret,Wallace2010,Chen2011,Jorswieck2013,Quist2013}, ultrawideband channels \cite{Wilson2007}, and on the orthogonal frequency division multiplexing modulation \cite{Liu2013,Wu2013,Zhang2019,Melki2019}. The authors in \cite{Liu2012} analytically compare RSS and CSI approaches. The work of \cite{Zhang2016} also compares the two approaches relying on a thorough experimental study in various propagation environments, with different degrees of mobility.

The majority of works in the literature considers that Eve gets no side information about the key from her observations, which consist of the pilots transmitted by Alice and Bob \cite{AzimiSadjadi2007,wong2009secret,Ye2010,Chen2011,Jorswieck2013}. Often, this assumption is justified by the fact that: (i) Eve is supposed to be separated from Bob and Alice by more than one wavelength (otherwise she could be easily detected) and (ii) the channel environment is supposed to be rich enough in scattering implying that the fading process of the channels at the different antennas can be considered independent. The assumption of rapid decorrelation in space has been further validated through measurement campaigns \cite{AzimiSadjadi2007,Mathur2008,Madiseh2009,Ye2010,Zhang2016}. Moreover, this assumption simplifies the expression of the secret-key capacity, which simply becomes equal to the mutual information between Alice and Bob. However, it often occurs in practical scenarios that scatterers are clustered with small angular spread rather than being uniformly distributed, which leads to much longer spatial decorrelation length, as was shown in \cite{rottenberg2020impact} for practical 3GPP channel models. In \cite{Chou2010}, the authors studied the impact of channel sparsity, inducing correlated eavesdropping, on the secret-key capacity. In \cite{Wallace2009}, the impact of the number of paths and the eavesdropper separation is analytically studied. In \cite{Zhang2017}, spatial and time correlation of the channel is taken into account using a Jakes Doppler model. In \cite{Pierrot2013,pierrot2013journal}, experiments are conducted indoor to evaluate the correlation of the eavesdropper's observations and its impact on the secret-key capacity. A similar study is conducted for a MIMO indoor measurement campaign in \cite{Wallace2010}. The work of \cite{Zenger2016} also uses an indoor experimental approach and proposes results of cross-correlation, mutual information and secret-key rates, which depend on the eavesdropper's position. 


\subsection{Contributions}

Our main contribution is to propose a novel analytical comparison of the secret-key capacity based on RSS and CSI. As opposed to similar previous works such as \cite{Liu2012}, we do not assume that Eve's observations are uncorrelated. This more general case adds to the complexity of the study while remaining of practical importance. Moreover, the authors in \cite{Liu2012} could characterize the secret-key capacity for envelope sampling with a simple analytical expression. However, their simplification relied on the approximation of a sum of envelope components as Gaussian, which is not applicable for our channel model. Furthermore, other works have already compared RSS and CSI-based approaches taking into account correlated eavesdropping, such as \cite{Zhang2016}. However, the studies were mostly conducted experimentally and not analytically.

More specifically, our contributions can be summarized as follows: 1) We evaluate lower and upper bounds on the secret-key capacity for both the complex (full CSI) and the envelope (RSS) cases. In the complex case, we obtain simple closed-form expressions, while, in the envelope case, the bounds must be evaluated numerically. Some of the expressions in the complex case were already obtained in previous works. We chose to present them again in this work to provide a systematic framework and useful comparison benchmarks for the envelope case.  2) We show that, in a number of particular cases, the lower and upper bounds become tight: low correlation of the eavesdropper, relatively smaller noise variance at Bob than Alice (and vice versa) and specific high signal-to-noise ratio (SNR) regimes. 3) We show that, as soon as Alice (or Bob since everything is symmetrical) samples the envelope of her channel estimate, the other parties do not loose information by taking the envelopes of their own channel estimates. 4) We show that, in the high SNR regime, the bounds can be evaluated in closed-form and result in simple expressions. The penalty of envelope-based versus complex-based secret-key generation is: i) a pre-log factor of $1/2$ instead of $1$, implying a slower slope of the secret-key capacity as a function of SNR and ii) a constant penalty of $0.69$ bit, which disappears as Eve's channel gets highly correlated.

The rest of this paper is structured as follows. Section~\ref{section:transmission_model} describes the transmission model used in this work. Sections~\ref{section:secret_key_capacity_complex} and \ref{section:secret_key_capacity_envelope} study the secret-key capacity based on complex and envelope sampling, respectively. Section~\ref{section:Numerical_validation} numerically validates the theoretical results. Finally, Section~\ref{section:conclusion} concludes the paper.


\subsection*{Notations}
Matrices are denoted by bold uppercase letters. Non bold upper case letter refers to a random variable. Superscript $^*$ stands for conjugate operator. The symbol $\Re(.)$ denotes the real part. $\jmath$ is the imaginary unit. $|\mat{A}|$ is the determinant of matrix $\mat{A}$. The letters $e$ and $\gamma$ refer to the Euler number and the Euler-Mascheroni constant respectively. $h(.)$ and $I(.;.)$ refer to the differential entropy and the mutual information respectively. We use the notation $f(x)=O(g(x))$, as $x\rightarrow a$, if there exist positive numbers $\delta$ and $\lambda$ such that $|f(x)|\leq \lambda g(x)$ when $0<|x-a|<\delta$.



\section{Transmission Model}
\label{section:transmission_model}

\begin{figure}[!t]  
	\centering

\resizebox{0.45\textwidth}{!}{%
	{\includegraphics[clip, trim=0cm 14cm 23cm 0cm, scale=1]{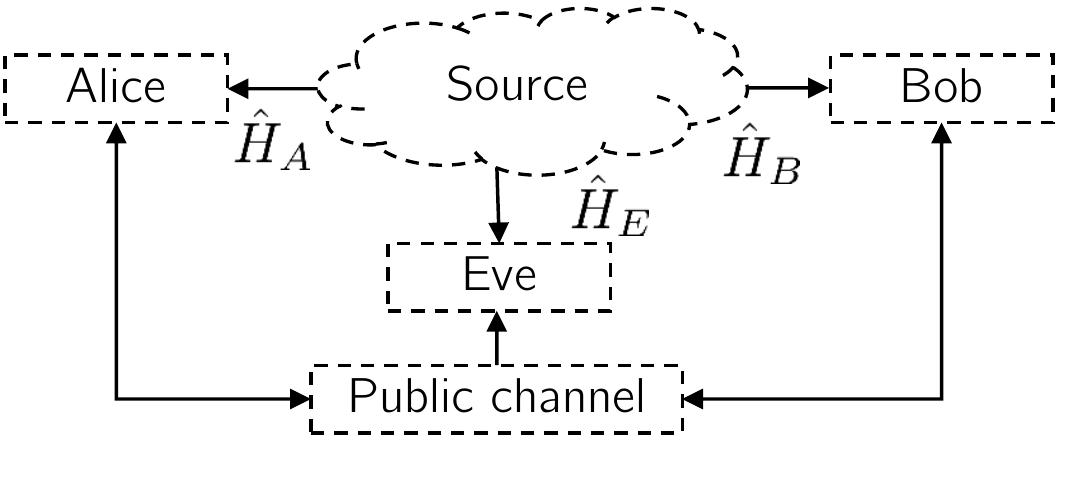}} 
}
	\vspace{-2em}
	\caption{Source model for secret-key agreement.}
	\label{fig:source_model}
	\vspace{-1em}
\end{figure}

Alice and Bob extract a common key from observations of their shared channel $H$, assumed to be reciprocal. The channel $H$ is repeatedly sampled in time based on the transmission of \textit{a priori} known pilots by Alice and Bob. We assume that the successive observations of $H$ are distant enough in time so that they can be considered independent. During these successive observations, the environment remains stationary so that they can be considered as identically distributed. Considering a narrowband channel, the estimates of $H$ at Alice's and Bob's sides, respectively denoted by $\hat{H}_A$ and $\hat{H}_B$, are given by
\begin{align*}
\hat{H}_A&=H+W_A,\ \hat{H}_B=H+W_B,
\end{align*}
where the additive noise samples $W_A$ and $W_B$ are modeled as independent zero mean circularly-symmetric complex Gaussian (ZMCSCG) random variables with variances $\sigma^2_A$ and $\sigma^2_B$ respectively.  

The strategy of Eve consists in going as close as possible from Bob's antenna to try to maximize the correlation of its channel\footnote{Note that all of the following derivations are symmetrical if Eve gets close to Alice instead of Bob.}. Then, Eve estimates her channel $H_E$ between Alice's antenna and hers by intercepting the pilots sent from Alice to Bob. Since Eve is close to Bob, the channel from Alice to Eve will be spatially correlated with $H$ while the channel between Bob and Eve will experience a negligible correlation with $H$. Therefore, we neglect the pilot sent by Bob and received by Eve in the following as she cannot get any useful information from it \cite{Wallace2009}. The channel estimate of Eve is given by
\begin{align*}
\hat{H}_E=H_E+W_E,
\end{align*}
where $W_E$ is modeled as ZMCSCG with variance $\sigma^2_E$. 
If Alice and Bob transmit a pilot of equal power and Alice, Bob and Eve use a similar receiver, one could expect a situation of equal noise variance $\sigma^2_A=\sigma^2_B=\sigma^2_E$. On the other hand, Eve could use a more powerful receiver than Alice and/or Bob by having, \textit{e.g.}, a larger antenna size, a multi-antenna receiver or an amplifier with lower noise figure. This would result in a lower noise variance $\sigma^2_E$. Moreover, a different pilot power transmitted by Alice and Bob will induce variations in their noise variances $\sigma^2_A$ and $\sigma^2_B$. Indeed, in practice, the channel estimates $\hat{H}_A$, $\hat{H}_B$ and $\hat{H}_E$ are obtained by dividing the received signal, which includes the additive noise, by an \textit{a priori} known pilot. For instance, if the pilot transmitted by Bob has a stronger power, the noise power at Alice $\sigma_A^2$ will be relatively weaker.

This scenario corresponds to the memoryless source model for secret-key agreement \cite{Csiszar1993,bloch2011physical} represented in Fig.~\ref{fig:source_model}: Alice, Bob and Eve observe a set of independent and identically distributed (i.i.d.) repetitions of the random variables $\hat{H}_A$, $\hat{H}_B$ and $\hat{H}_E$. 
Moreover, an error-free authenticated public channel of unlimited capacity is available for communication. All parties have access to the public channel.

In the following section, we will study the secret-key capacity of this model. To do this, we need to know the probability distributions of the random variables $\hat{H}_A$, $\hat{H}_B$ and $\hat{H}_E$, which directly depend on the probability distributions of $W_A$, $W_B$, $W_E$, $H$ and $H_E$. The distributions of $W_A$, $W_B$ and $W_E$ were already detailed. Moreover, measurement campaigns have shown that the channels $H$ and $H_E$ can be accurately modeled with a ZMCSCG distribution, especially in non-line-of-sight situations and rich scattering environments \cite{durgin2003space}. Therefore, we assume that $(H, H_E)$ follows a ZMCSCG with covariance matrix given by
\begin{align*}
\mat{C}_{HH_E}=p\begin{pmatrix}
1& \rho\\
\rho^*& 1
\end{pmatrix},
\end{align*}
where $p$ is the channel variance, such that $0<p<\infty$ and $\rho$ is the spatial correlation coefficient, such that $0\leq |\rho| \leq 1$. We assume that $H$ and $H_E$ have the same variance $p$, which makes sense in practice if Bob and Eve are close enough so as to belong to the same local area \cite{durgin2003space}. In the following, we use the fact the differential entropy of a circularly symmetric Gaussian with covariance $\mat{C}$ is given by $\log_2( |\pi e\mat{C}|)$, where $e$ is the Euler number.

In the sequel, at different places, we will consider the high SNR regime. When this regime is considered, we will always assume, implicitly or explicitly, that, as $\sigma_A^2\rightarrow 0$, $\sigma_B^2\rightarrow 0$ and $\sigma_E^2\rightarrow 0$,

$\mathbf{(As1)}$: the ratio $\frac{\sigma_A^2}{\sigma_B^2}$ remains fixed and $0<\frac{\sigma_A^2}{\sigma_B^2}<\infty$,

$\mathbf{(As2)}$: the ratio $\frac{\sigma_A^2}{\sigma_E^2}$ remains fixed and $0<\frac{\sigma_A^2}{\sigma_E^2}<\infty$,

$\mathbf{(As3)}$: the ratio $\frac{\sigma_B^2}{\sigma_E^2}$ remains fixed and $0<\frac{\sigma_B^2}{\sigma_E^2}<\infty$.

%
%




\section{Secret-Key Capacity based on Complex Channel Sampling}
\label{section:secret_key_capacity_complex}

In this section, we analyze the secret-key capacity associated with complex channel sampling, that we denote by $C_{s}^{\mathrm{Cplex}}$. Some of the results were already derived in previous works. Moreover, most of them result from a direct evaluation of standard formulas for the differential entropy of Gaussian random variables. We still present them as they provide accurate benchmarks as a comparison with the envelope case presented in Section~\ref{section:secret_key_capacity_envelope}.

The secret-key capacity is defined as the maximal rate at which Alice and Bob can agree on a secret-key while keeping the rate at which Eve obtains information about the key arbitrarily small for a sufficiently large number of observations. Moreover, Alice and Bob should agree on a common key with high probability and the key should approach the uniform distribution. We refer to \cite{Csiszar1993,Maurer1993,bloch2011physical} for a formal definition. As explained in Section~\ref{section:transmission_model}, we consider that Eve gets useful information from her observation $\hat{H}_E$ over $H$. This implies that the secret-key capacity is not simply equal to $I(\hat{H}_A;\hat{H}_B)$, as was considered in many previous works \cite{AzimiSadjadi2007,Chen2011,Jorswieck2013,Ye2006,Ye2010}. Finding the general expression of the secret-key capacity for a given probability distribution of $\hat{H}_A,\hat{H}_B,\hat{H}_E$ is still an open problem. From \cite{Csiszar1993,Maurer1993} \cite[Prop. 5.4]{bloch2011physical}, the secret-key capacity, expressed in the number of generated secret bits per channel observation, can be lower and upper bounded as follows
\begin{align}
C_{s}^{\mathrm{Cplex}}&\geq  I(\hat{H}_A;\hat{H}_B)-\min\left[I(\hat{H}_A;\hat{H}_E), I(\hat{H}_B;\hat{H}_E) \right] \label{eq:lower_bounds}\\
C_{s}^{\mathrm{Cplex}}&\leq \min\left[ I(\hat{H}_A;\hat{H}_B),I(\hat{H}_A;\hat{H}_B|\hat{H}_E) \right].  \label{eq:upper_bounds}
\end{align}
The lower bound (\ref{eq:lower_bounds}) implies that, if Eve has less information about $\hat{H}_B$ than Alice or respectively about $\hat{H}_A$ than Bob, such a difference can be leveraged for secrecy \cite{Maurer1993}. Moreover, this rate can be achieved with one-way communication. On the other hand, the upper bound (\ref{eq:upper_bounds}) implies that the secret-key rate cannot exceed the mutual information between Alice and Bob. Moreover, the secret-key rate cannot be higher than the mutual information between Alice and Bob if they happened to learn Eve's observation $\hat{H}_E$. In particular cases, the lower and upper bounds can become tight. In our context, three particular cases can be distinguished:

\begin{enumerate}
	\item $\rho=0$: Eve does not learn anything about $H$ from $\hat{H}_E$, which becomes independent from $\hat{H}_A$ and $\hat{H}_B$. This leads to the trivial result $C_{s}^{\mathrm{Cplex}}=I(\hat{H}_A;\hat{H}_B)$.
	
	\item $\sigma^2_B=0$: this implies that $\hat{H}_A\rightarrow \hat{H}_B \rightarrow \hat{H}_E$ forms a Markov chain, which leads to
		$C_{s}^{\mathrm{Cplex}}=I(\hat{H}_A;\hat{H}_B|\hat{H}_E)=I(\hat{H}_A;\hat{H}_B)-I(\hat{H}_A;\hat{H}_E)$ \cite[Corol. 4.1]{bloch2011physical}.
	
	\item $\sigma^2_A=0$: symmetrically as in 2),
		$C_{s}^{\mathrm{Cplex}}=I(\hat{H}_A;\hat{H}_B|\hat{H}_E)=I(\hat{H}_A;\hat{H}_B)-I(\hat{H}_B;\hat{H}_E)$.
\end{enumerate}
Cases 2) and 3) can be verified in practice as the receiver of Alice or Bob is significantly noisier than the one of the other: as $\sigma_A^2\rightarrow 0$ for a fixed value of $\sigma_B^2$ or as $\sigma_B^2\rightarrow 0$ for a fixed value of $\sigma_A^2$. This would happen for instance if the power of Alice pilot is much stronger than the one of Bob.

In the next subsections, we evaluate the different expressions of the mutual information required to compute the lower and upper bounds of (\ref{eq:lower_bounds}) and (\ref{eq:upper_bounds}): i) the mutual information between Alice and Bob $I(\hat{H}_A;\hat{H}_B)$; ii) the mutual information between Alice and Eve $I(\hat{H}_A;\hat{H}_E)$, and similarly for Bob $I(\hat{H}_B;\hat{H}_E)$; and iii) the conditional mutual information between Alice and Bob given Eve's observations $I(\hat{H}_A;\hat{H}_B|\hat{H}_E)$.


\begin{figure*}[!t]  
	{\small \begin{align}
	C_{s}^{\mathrm{Cplex}}&\geq \log_2\left(1+\frac{p}{\sigma^2_A+\sigma^2_B+\frac{\sigma^2_A\sigma^2_B}{p}}\right)-\log_2\left(1+\frac{p|\rho|^2}{p(1-|\rho|^2)+\max(\sigma^2_A,\sigma^2_B)+\sigma^2_E+\frac{\max(\sigma^2_A,\sigma^2_B)\sigma^2_E}{p}}\right)\label{eq:LB_complex}.	\\
	C_{s}^{\mathrm{Cplex}}&\leq \log_2\left(\frac{\left[(p+\sigma_A^2)(p+\sigma_E^2)-|\rho p|^2\right]\left[(p+\sigma_B^2)(p+\sigma_E^2)-|\rho p|^2\right]}{(p+\sigma^2_E)\left[(p(\sigma^2_A+\sigma^2_B)+\sigma^2_A\sigma^2_B)(p+\sigma_E^2)-|\rho p|^2(\sigma_A^2+\sigma_B^2)\right]}\right). \label{eq:UB_complex}
	\end{align}	}	\vspace{-2em}
\end{figure*}

\subsection{Mutual Information between Alice and Bob}

Using previously introduced transmission and channel models, we can find that the random variables $\hat{H}_A$ and $\hat{H}_B$ are jointly Gaussian distributed with covariance
\begin{align*}
\mat{C}_{\hat{H}_A\hat{H}_B}=\begin{pmatrix}
p+\sigma^2_A& p\\
p& p+\sigma^2_B
\end{pmatrix}.
\end{align*}
From this distribution, we find back the result of \cite{Ye2006}
\begin{align}
I(\hat{H}_A;\hat{H}_B)&=h(\hat{H}_A)+h(\hat{H}_B)-h(\hat{H}_A,\hat{H}_B)\label{eq:I_XY}\\
	&=\log_2\left( \frac{(p+\sigma^2_A)(p+\sigma^2_B)}{|\mat{C}_{\hat{H}_A\hat{H}_B}|}\right)\nonumber\\
&=\log_2\left(1+\frac{p}{\sigma^2_A+\sigma^2_B+\frac{\sigma^2_A\sigma^2_B}{p}}\right)
.\nonumber
\end{align}
This rate corresponds to the secret-key capacity in case of uncorrelated observations at Eve ($\rho=0$). At high SNR, as $\sigma_A^2\rightarrow 0$ and $\sigma_B^2\rightarrow 0$, the expressions becomes 
\begin{align}
I(\hat{H}_A;\hat{H}_B)= \log_2\left(\frac{p}{\sigma^2_A+\sigma^2_B}\right)+O\left(\sigma_A^2\right), \label{eq:I_HA_HB_complex_high_SNR}
\end{align}
which is characterized by a \textit{pre-log factor} of one.

\subsection{Mutual Information between Alice/Bob and Eve}

We can observe that $\hat{H}_A$ and $\hat{H}_E$ are jointly Gaussian distributed with covariance
\begin{align*}
\mat{C}_{\hat{H}_A\hat{H}_E}=\begin{pmatrix}
p+\sigma^2_A& \rho p\\
\rho^*p& p+\sigma^2_E
\end{pmatrix}.
\end{align*}
This leads to the mutual information
\begin{align*}
I(\hat{H}_A;\hat{H}_E)
&=\log_2\left( \frac{(p+\sigma^2_A)(p+\sigma^2_E)}{|\mat{C}_{\hat{H}_A\hat{H}_E}|}\right)\nonumber
=\log_2\left(1+\frac{p|\rho|^2}{p(1-|\rho|^2)+\sigma^2_A+\sigma^2_E+\frac{\sigma^2_A\sigma^2_E}{p}}\right).
\end{align*}
The mutual information $I(\hat{H}_B;\hat{H}_E)$ can be similarly obtained, simply replacing subscript $A$ by $B$. Using the previously derived expressions of $I(\hat{H}_A;\hat{H}_B)$, $I(\hat{H}_A;\hat{H}_E)$ and $I(\hat{H}_B;\hat{H}_E)$, we find that the lower bound in (\ref{eq:lower_bounds}) evaluates to (\ref{eq:LB_complex}). Note that the lower bound is not restricted to be positive (as will also be shown numerically in Section~\ref{section:Numerical_validation}), in which case it becomes useless since, by definition, $C_{s}^{\mathrm{Cplex}}\geq 0$. Nonetheless, it does not necessarily imply that $C_{s}^{\mathrm{Cplex}}=0$. We can find the condition on the minimum noise variance at Eve $\sigma^2_E$ for having a larger-than-zero lower bound
\begin{align}
\sigma^2_E&> p(|\rho |^2-1) +|\rho |^2\min(\sigma^2_A,\sigma^2_B). \label{eq:lower_bound_N_Z}
\end{align}
In the worst-case, $|\rho|=1$ and $\sigma^2_E$ has to be larger than the minimum of the noise variances of Alice and Bob. We can invert (\ref{eq:lower_bound_N_Z}) to find the maximal correlation coefficient $|\rho|^2$ to have a larger-than-zero lower bound
\begin{align*}
|\rho |^2&< \frac{p+\sigma^2_E}{p+\min(\sigma^2_A,\sigma^2_B)}.
\end{align*}
In the high SNR regime, as $\sigma_A^2\rightarrow 0$, $\sigma_B^2\rightarrow 0$ and $\sigma_E^2\rightarrow 0$, equation (\ref{eq:LB_complex}) becomes
\begin{align}
&C_{s}^{\mathrm{Cplex}}\geq \log_2\left(\frac{p}{\sigma^2_A+\sigma^2_B}\right)-\log_2\left(\frac{p}{p(1-|\rho|^2)+\max(\sigma_A^2,\sigma_B^2)+\sigma_E^2}\right)+O\left(\sigma_A^2\right)\label{eq:asymptotic_LB_complex}
.
\end{align}
As soon as $|\rho|<1$, $C_{s}^{\mathrm{Cplex}}$ is unbounded and goes to infinity at high SNR. Indeed, 
\begin{align*}
	&\lim_{\sigma_A^2,\sigma_B^2\rightarrow 0} I(\hat{H}_A;\hat{H}_B)= +\infty, \ \lim_{\sigma_A^2,\sigma_E^2\rightarrow 0} I(\hat{H}_A;\hat{H}_E)=\lim_{\sigma_B^2,\sigma_E^2\rightarrow 0} I(\hat{H}_B;\hat{H}_E)=\log_2\left(\frac{1}{1-|\rho|^2}\right), 
\end{align*}
which is bounded for $|\rho|<1$. 

\subsection{Conditional Mutual Information between Alice and Bob} \label{subsection:conditional_I_complex}
We can note that $\hat{H}_A$, $\hat{H}_B$ and $\hat{H}_E$ are jointly Gaussian distributed with covariance matrix
\begin{align*}
\mat{C}_{\hat{H}_A\hat{H}_B\hat{H}_E}=\begin{pmatrix}
p+\sigma^2_A& p & \rho p\\
p& p+\sigma^2_B & \rho p\\
\rho^* p& \rho^* p& p+\sigma^2_E
\end{pmatrix},
\end{align*}
which gives
\begin{align}
I(\hat{H}_A;\hat{H}_B|\hat{H}_E)
&=h(\hat{H}_A,\hat{H}_E)-h(\hat{H}_E)+h(\hat{H}_B,\hat{H}_E)-h(\hat{H}_A,\hat{H}_B,\hat{H}_E)\label{eq:def_joint_diff_entropy3}\\
&=\log_2\left(\frac{|\mat{C}_{\hat{H}_A\hat{H}_E}||\mat{C}_{\hat{H}_B\hat{H}_E}|}{(p+\sigma^2_E)|\mat{C}_{\hat{H}_A\hat{H}_B\hat{H}_E}|}\right).\nonumber
\end{align}

The upper bound in (\ref{eq:upper_bounds}) is then given by the minimum of $I(\hat{H}_A;\hat{H}_B|\hat{H}_E)$ and $I(\hat{H}_A;\hat{H}_B)$. In Appendix~\ref{subsection:proof1}, we prove that the condition $I(\hat{H}_A;\hat{H}_B|\hat{H}_E)\leq I(\hat{H}_A;\hat{H}_B)$ is always verified under the jointly Gaussian channel model considered in this work. The upper bound is thus given by (\ref{eq:UB_complex}).

Based on the analytical expressions of the upper and lower bounds, we can find a novel condition for tightness of the bounds at high SNR.
\begin{proposition} \label{proposition:tight_high_SNR}
	Under $\mathbf{(As1)}-\mathbf{(As3)}$, as $\sigma_A^2\rightarrow 0$, $\sigma_B^2\rightarrow 0$ and $\sigma_E^2\rightarrow 0$, if $|\rho|<1$, the upper and low bounds of (\ref{eq:LB_complex}) and (\ref{eq:UB_complex}) become tight and the secret-key capacity is given by
	\begin{align} 
		C_{s}^{\mathrm{Cplex}}= \log_2\left(\frac{p(1-|\rho|^2)}{\sigma_A^2+\sigma_B^2}\right) +O\left(\sigma_A^2\right). \label{eq:prop_tight_high_SNR}
	\end{align}
\end{proposition}
\begin{proof}
	The proof is easily obtained by taking the limits in (\ref{eq:LB_complex}) and (\ref{eq:UB_complex}) and seeing that they both converge towards (\ref{eq:prop_tight_high_SNR}), provided that $|\rho|<1$.
\end{proof}

\section{Secret-Key Capacity based on Channel Envelope Sampling}
\label{section:secret_key_capacity_envelope}

The goal of this section is to evaluate the impact on the secret-key capacity if Alice and Bob rely on the envelopes of their observations rather than the complex values to generate a secret key. We denote by $C_s^{\mathrm{Evlpe}}$ the secret-key capacity based on envelope sampling. We also introduce the notations 
\begin{align*}
	\hat{H}_A=\hat{R}_A e^{\jmath\hat{\Phi}_A},\ \hat{H}_B=\hat{R}_B e^{\jmath\hat{\Phi}_B},\ \hat{H}_E=\hat{R}_E e^{\jmath\hat{\Phi}_E},
\end{align*}
where $\hat{R}_A$, $\hat{R}_B$ and $\hat{R}_E$ are the random modules of $\hat{H}_A$, $\hat{H}_B$ and $\hat{H}_E$ respectively. Similarly, $\hat{\Phi}_A$, $\hat{\Phi}_B$ and $\hat{\Phi}_E$ are their random phases. Note that $\hat{H}_A$ is equivalently represented by $\hat{R}_A$ and $\hat{\Phi}_A$ or $\Re(\hat{H}_A)$ and $\Im(\hat{H}_A)$. We start by stating an insightful result from \cite[Th.~2]{Liu2012}, that we generalize for Eve's observations.

\begin{proposition} \label{theorem:liu} The mutual information $I(\hat{H}_A;\hat{H}_B)$ satisfies
	\begin{align*}
	I(\hat{H}_A;\hat{H}_B)&=I(\Re(\hat{H}_A);\Re(\hat{H}_B))+I(\Im(\hat{H}_A);\Im(\hat{H}_B))\\
	&\geq 	I(\hat{R}_A;\hat{R}_B)+I(\hat{\Phi}_A;\hat{\Phi}_B).
	\end{align*}
	\begin{proof} 
		The proof is obtained as a particular case of Proposition~\ref{theorem:liu_Eve} for $\rho=1$ and replacing subscripts $E$ by $B$.		
	\end{proof}
\end{proposition}

\begin{proposition} \label{theorem:liu_Eve} The mutual information $I(\hat{H}_A;\hat{H}_E)$ satisfies
	\begin{align*}
		I(\hat{H}_A;\hat{H}_E)&=I(\Re(\hat{H}_A);\Re(\hat{H}_E))+I(\Im(\hat{H}_A);\Im(\hat{H}_E))\\
		&\geq 	I(\hat{R}_A;\hat{R}_E)+I(\hat{\Phi}_A;\hat{\Phi}_E).
	\end{align*}
	\begin{proof} We here give the main lines of the proof to give the general intuition to the reader. Complements for the dependence of random variables are given in Appendix~\ref{subsection:proof_RA_RE}. On the one hand, we have
		\begin{align*}
		I(\hat{H}_A;\hat{H}_E)&=I(\hat{R}_A,\hat{\Phi}_A;\hat{R}_E,\hat{\Phi}_E)\\
		&=h(\hat{R}_A,\hat{\Phi}_A)-h(\hat{R}_A,\hat{\Phi}_A|\hat{R}_E,\hat{\Phi}_E)\\
		&\stackrel{(*)}{=}h(\hat{R}_A)-h(\hat{R}_A|\hat{R}_E,\hat{\Phi}_E)+h(\hat{\Phi}_A)-h(\hat{\Phi}_A|\hat{R}_A,\hat{R}_E,\hat{\Phi}_E)\\
		&\stackrel{(**)}{\geq}I(\hat{R}_A;\hat{R}_E)+I(\hat{\Phi}_A;\hat{\Phi}_E),
		\end{align*}
		where $(*)$ follows from the chain rule for entropy and the fact that $\hat{R}_A$ and $\hat{\Phi}_A$ are independent since the envelope and the phase of a ZMCSG are independent. $(**)$ follows from the fact that: i) $h(\hat{R}_A|\hat{R}_E,\hat{\Phi}_E)=h(\hat{R}_A|\hat{R}_E)$ since $(\hat{R}_A,\hat{R}_E)$ and $\hat{\Phi}_E$ are independent; ii) $h(\hat{\Phi}_A|\hat{R}_A,\hat{R}_E,\hat{\Phi}_E)\geq h(\hat{\Phi}_A|\hat{\Phi}_E)$ by the general properties of differential entropy and since $(\hat{\Phi}_A,\hat{\Phi}_E)$ is not independent from $(\hat{R}_A,\hat{R}_E)$.
		
		On the other hand, a similar derivation can be made for $I(\Re(\hat{H}_A),\Im(\hat{H}_A);\Re(\hat{H}_E),\Im(\hat{H}_E))$, noticing that $\hat{H}_A$ and $\hat{H}_E$ are two ZMCSG, implying that their real and imaginary parts are independent, resulting in an equality with $I(\hat{H}_A;\hat{H}_E)$.		
	\end{proof}
\end{proposition}
Intuitively, this result can be explained by the fact that the random vectors $(\hat{\Phi}_A,\hat{\Phi}_E)$ and $(\hat{R}_A,\hat{R}_E)$ are not independent from one another while $(\Re(\hat{H}_A),\Re(\hat{H}_E))$ and $(\Im(\hat{H}_A),\Im(\hat{H}_E))$ are. There is thus a loss of information by treating phase and envelope separately as opposed to real and imaginary parts.

One could wonder what is the best strategy of Bob and Eve if Alice uses $\hat{R}_A$ to generate a key. Imagine Bob and Eve have a more advanced receiver so that they can sample their observations in the complex domain, would it be beneficial for them? The answer is no, as shown in the two following propositions.

\begin{proposition} \label{proposition:No_loss_info_envelope_Bob}
	If Alice uses the envelope of her observations $\hat{R}_A$, then Bob does not loose information by taking the envelope of $\hat{H}_B$, \textit{i.e.}, $I(\hat{R}_A;\hat{H}_B)=I(\hat{R}_A;\hat{R}_B)$. The same result holds if Alice and Bob's roles are interchanged.
\end{proposition}
\begin{proof}
	The proof is obtained as a particular case of Proposition~\ref{proposition:No_loss_info_envelope_Eve} for $\rho=1$ and replacing subscripts $E$ by $B$.
\end{proof}
\begin{proposition} \label{proposition:No_loss_info_envelope_Eve}
	If Alice uses the envelope of her observations $\hat{R}_A$, then Eve does not loose information by taking the envelope of $\hat{H}_E$, \textit{i.e.}, $I(\hat{R}_A;\hat{H}_E)=I(\hat{R}_A;\hat{R}_E)$. The same result holds for Bob's observations.
\end{proposition}
\begin{proof}
	Here again, we refer to Appendix~\ref{subsection:proof_RA_RE} for the complementary proofs on dependence of random variables. By definition, we have
	\begin{align*}
	I(\hat{R}_A;\hat{R}_E,\hat{\Phi}_E)&=h(\hat{R}_E,\hat{\Phi}_E)-h(\hat{R}_E,\hat{\Phi}_E|\hat{R}_A)\\
	&\stackrel{(*)}{=} h(\hat{R}_E)+h(\hat{\Phi}_E)-\left(h(\hat{R}_E|\hat{R}_A)+h(\hat{\Phi}_E|\hat{R}_A,\hat{R}_E)\right)\\
	&\stackrel{(**)}{=}I(\hat{R}_A;\hat{R}_E),
	\end{align*}
	where $(*)$ relies on the chain rule for entropy and the fact that $\hat{R}_E$ and $\hat{\Phi}_E$ are independent since the envelope and the phase of a ZMCSG are independent. $(**)$ relies on the fact that: i) $h(\hat{\Phi}_E|\hat{R}_E)=h(\hat{\Phi}_E)$ since envelope and the phase of a ZMCSG are independent and ii) $h(\hat{\Phi}_E|\hat{R}_A,\hat{R}_E)=h(\hat{\Phi}_E)$ since $(\hat{R}_A,\hat{R}_E)$ and $\hat{\Phi}_E$ are independent.
\end{proof}
Intuitively, the propositions can be explained by the fact that $\hat{\Phi}_B$ and $\hat{\Phi}_E$ are independent from $(\hat{R}_A,\hat{R}_B)$ and $(\hat{R}_A,\hat{R}_E)$ respectively. The propositions provide practical insight in the sense that, as soon as Alice (or Bob since everything is symmetrical) samples the envelope of her channel estimate, the other parties do not loose information by taking the envelopes of their own channel estimates. The other way around, Bob or Eve would not gain information to work on their complex channel estimate. In the lights of this result, the definitions of the bounds of the secret-key capacity defined in (\ref{eq:lower_bounds}) and (\ref{eq:upper_bounds}) also hold here by replacing the complex values by their envelopes, \textit{i.e.}, $\hat{R}_A$, $\hat{R}_B$ and $\hat{R}_E$ instead of $\hat{H}_A$, $\hat{H}_B$ and $\hat{H}_E$ respectively:
\begin{align}
C_s^{\mathrm{Evlpe}}&\geq  I(\hat{R}_A;\hat{R}_B)-\min\left[I(\hat{R}_A;\hat{R}_E), I(\hat{R}_B;\hat{R}_E) \right]\label{eq:lower_bounds_R}\\
C_s^{\mathrm{Evlpe}}&\leq \min\left[ I(\hat{R}_A;\hat{R}_B),I(\hat{R}_A;\hat{R}_B|\hat{R}_E) \right].  \label{eq:upper_bounds_R}
\end{align}

Tight bounds can be found in the same cases and for the same reasons as in the complex case: 1) $\rho=0$, 2) $\sigma^2_B=0$ and 3) $\sigma^2_A=0$.

Similarly as in Section~\ref{section:secret_key_capacity_complex}, we evaluate in the following subsections the quantities required to compute the lower and upper bounds (\ref{eq:lower_bounds_R}) and (\ref{eq:upper_bounds_R}): i) the mutual information between Alice and Bob $I(\hat{R}_A;\hat{R}_B)$; ii) the mutual information between Alice and Eve $I(\hat{R}_A;\hat{R}_E)$, and similarly for Bob $I(\hat{R}_B;\hat{R}_E)$; and iii) the conditional mutual information between Alice and Bob given Eve's observations $I(\hat{R}_A;\hat{R}_B|\hat{R}_E)$.

\subsection{Mutual Information between Alice and Bob}

The mutual information between Alice and Bob is given by
\begin{align}
I(\hat{R}_A;\hat{R}_B)&=h(\hat{R}_A)+h(\hat{R}_B)-h(\hat{R}_A,\hat{R}_B) \label{eq:def_I_RA_RB}.
\end{align}
The envelope of a ZMCSG random variable is well known to be Rayleigh distributed, \textit{i.e.}, $\hat{R}_A\sim \text{Rayleigh}(\sqrt{\frac{p+\sigma_A^2}{2}})$ and $\hat{R}_B\sim \text{Rayleigh}(\sqrt{\frac{p+\sigma_B^2}{2}})$. The differential entropy of a Rayleigh distribution is also well known and is equal to \cite{michalowicz2013handbook}
\begin{align}
h(\hat{R}_A)&= \frac{1}{2}\log_2 \left(\frac{{p+\sigma_A^2}}{4}\right) + \frac{1}{2}\log_2 (e^{2+\gamma})\label{eq:h_RA}\\
h(\hat{R}_B)&= \frac{1}{2}\log_2 \left(\frac{{p+\sigma_B^2}}{4}\right) + \frac{1}{2}\log_2 (e^{2+\gamma})\label{eq:h_RB},
\end{align}
where $\gamma$ is the Euler-Mascheroni constant and $e$ is the Euler number. On the other hand, the joint differential entropy of $(\hat{R}_A,\hat{R}_B)$ is more difficult to compute. The following lemma gives the joint probability density function (PDF) of $(\hat{R}_A,\hat{R}_B)$.

\begin{lemma} \label{lemma:joint_dist_RA_RB}
	The joint PDF of $(\hat{R}_A,\hat{R}_B)$ is given by
	\begin{align*}
	f_{\hat{R}_A,\hat{R}_B}(\hat{r}_A,\hat{r}_B)=&\frac{4 \hat{r}_A\hat{r}_B}{p(\sigma_A^2+\sigma_B^2)+\sigma_A^2\sigma_B^2}I_0\left(\frac{2p\hat{r}_A\hat{r}_B}{p(\sigma_A^2+\sigma_B^2)+\sigma_A^2\sigma_B^2}\right)\\
	&\exp\left({-\frac{\hat{r}_A^2(p+\sigma_B^2)+\hat{r}_B^2(p+\sigma_A^2)}{p(\sigma_A^2+\sigma_B^2)+\sigma_A^2\sigma_B^2}}\right),
	\end{align*}
	where $I_0(.)$ is the zero order modified Bessel function of the first kind.
\end{lemma}
\begin{proof}
		The proof is obtained as a particular case of Lemma~\ref{lemma:joint_dist_RA_RE} for $\rho=1$ and replacing subscripts $E$ by $B$.
\end{proof}

Unfortunately, finding a closed-form expression for the joint differential entropy $h(\hat{R}_A,\hat{R}_B)$ is non-trivial given the presence of the Bessel function \cite{michalowicz2013handbook}. Still, $h(\hat{R}_A,\hat{R}_B)$ and thus $I(\hat{R}_A;\hat{R}_B)$, can be evaluated by numerical integration, relying on the PDF obtained in Lemma~\ref{lemma:joint_dist_RA_RB}. 

In the high SNR regime, the following lemma shows the limiting behavior of the PDF $f_{\hat{R}_A,\hat{R}_B}(\hat{r}_A,\hat{r}_B)$, which can be used to obtain a simple closed-form expression of $I(\hat{R}_A;\hat{R}_B)$, as shown in the subsequent theorem.

\begin{lemma} \label{lemma:joint_dist_RA_RB_high_SNR}
	Under $\mathbf{(As1)}$, as $\sigma_A^2\rightarrow 0$ and $\sigma_B^2\rightarrow 0$, the PDF $f_{\hat{R}_A,\hat{R}_B}(\hat{r}_A,\hat{r}_B)$ asymptotically converges to
	\begin{align*}
	f_{\hat{R}_A,\hat{R}_B}(\hat{r}_A,\hat{r}_B)&= \frac{2\hat{r}_Ae^{-\frac{\hat{r}_A^2}{p}}}{p}\frac{e^{-\frac{(\hat{r}_B-\hat{r}_A)^2 }{\sigma_A^2+\sigma_B^2}}}{\sqrt{\pi(\sigma_A^2+\sigma_B^2)}}+O\left({\sigma_A}\right),
	\end{align*}
	which corresponds to the product of a Rayleigh distribution of parameter $\frac{p}{2}$ and a conditional normal distribution centered in $\hat{R}_A$ and of variance $\frac{\sigma_A^2+\sigma_B^2}{2}$.
\end{lemma}
\begin{proof}
		The proof is obtained as a particular case of Lemma~\ref{lemma:joint_dist_RA_RE_high_SNR} for $\rho=1$ and replacing subscripts $E$ by $B$. Since $\rho=1$, the limit ${|\rho| \rightarrow 1}$ can be omitted.
\end{proof}

\begin{theorem} \label{theorem:I_RA_RB_high_SNR}
	Under $\mathbf{(As1)}$, as $\sigma_A^2\rightarrow 0$ and $\sigma_B^2\rightarrow 0$, the mutual information $I(\hat{R}_A;\hat{R}_B)$ converges to
	\begin{align*}
		I(\hat{R}_A;\hat{R}_B)\rightarrow& \frac{1}{2} \log_2\left(\frac{p}{\sigma_A^2+\sigma_B^2}\right)-\chi,
	\end{align*}
	where $\chi=\frac{1}{2} \log_2\left(\frac{4\pi}{e^{1+\gamma}}\right)$ is a constant penalty, given by $0.69$ (up to the two first decimals).
\end{theorem}
\begin{proof}
	The proof is obtained as a particular case of Theorem~\ref{theorem:I_RA_RE_high_SNR_correlation} for $\rho=1$ and replacing subscripts $E$ by $B$. Since $\rho=1$, the limit ${|\rho| \rightarrow 1}$ can be omitted.
\end{proof}
The expression obtained in Theorem~\ref{theorem:I_RA_RB_high_SNR} gives a lot of insight on the high SNR secret-key capacity that can be obtained with envelope sampling, when there is no correlation ($\rho=0$). As shown in the left column of Table~\ref{table:1}, two penalties can be observed as compared to complex sampling: i) a \textit{pre-log factor} of $1/2$ instead of $1$, implying a curve with smaller slope and ii) an additional penalty of a constant $\chi$ equivalent to about $0.69$ bit.

\begin{table*}[t!]
	\centering 
	\centering
	\resizebox{0.98\textwidth}{!}{%
		\renewcommand{\arraystretch}{1.1}
		\begin{tabular}{|c|c|c|} 
			\hline
			& High SNR ($\sigma_A^2,\sigma_B^2\rightarrow 0$), uncorrelated ($\rho=0$) & High SNR ($\sigma_A^2,\sigma_B^2,\sigma_E^2\rightarrow 0$), correlated ($|\rho|>0$) \\ \hline
			Complex & $C_{s}^{\mathrm{Cplex}}=\log_2\left(\frac{p}{\sigma_A^2+\sigma_B^2}\right)+O(\sigma_A^2)$ & $C_{s}^{\mathrm{Cplex}}\geq \log_2\left(\frac{p}{\sigma^2_A+\sigma^2_B}\right)-\log_2\left(\frac{p}{p(1-|\rho|^2)+\sigma_{*}^2+\sigma_E^2}\right)+O(\sigma_A^2)$ \\ \hline
			Envelope & $C_{s}^{\mathrm{Evlpe}}=\frac{1}{2}\log_2\left(\frac{p}{\sigma_A^2+\sigma_B^2}\right)-\chi+\epsilon_{\mathrm{uncrl}} $ & $C_{s}^{\mathrm{Evlpe}}\underset{|\rho|\rightarrow 1}{\geq} \frac{1}{2}\left[\log_2\left(\frac{p}{\sigma^2_A+\sigma^2_B}\right)-\log_2\left(\frac{p}{p(1-|\rho|^2)+\sigma_{*}^2+\sigma_E^2}\right)\right]+\epsilon_{\mathrm{crl}}  $ \\ \hline
		\end{tabular} 
	}
\vspace{0.5em}
	\caption{High SNR secret-key capacity of complex (CSI) versus envelope (RSS) sampling in both uncorrelated and correlated cases, under $\mathbf{(As1)}$-$\mathbf{(As3)}$. $\chi=0.69...$, $\sigma_{*}^2=\max(\sigma_A^2,\sigma_B^2)$, $\epsilon_{\mathrm{uncrl}}\rightarrow 0$, $\epsilon_{\mathrm{crl}}\rightarrow 0$ asymptotically.}
	\label{table:1}
\end{table*}

\subsection{Mutual Information between Alice/Bob and Eve}

We now analyze the mutual information between Alice and Eve and between Bob and Eve, which are given by
\begin{align}
I(\hat{R}_A;\hat{R}_E)&=h(\hat{R}_A)+h(\hat{R}_E)-h(\hat{R}_A,\hat{R}_E) \label{eq:def_I_RA_RE}\\
I(\hat{R}_B;\hat{R}_E)&=h(\hat{R}_B)+h(\hat{R}_E)-h(\hat{R}_B,\hat{R}_E).\nonumber
\end{align}
We already computed the values of $h(\hat{R}_A)$ and $h(\hat{R}_B)$. Similarly as for $\hat{R}_A$ and $\hat{R}_B$, we find that $\hat{R}_E\sim \text{Rayleigh}(\sqrt{\frac{p+\sigma_E^2}{2}})$ and \cite{michalowicz2013handbook}
\begin{align}
h(\hat{R}_E)&= \frac{1}{2}\log_2 \left(\frac{{p+\sigma_E^2}}{4}\right) + \frac{1}{2}\log_2 (e^{2+\gamma}).\label{eq:h_RE}
\end{align}
The following lemma gives the joint PDFs of $(\hat{R}_A,\hat{R}_E)$ and $(\hat{R}_B,\hat{R}_E)$.

\begin{lemma} \label{lemma:joint_dist_RA_RE}
	The joint PDF of $(\hat{R}_A,\hat{R}_E)$ is given by
	\begin{align*}
	f_{\hat{R}_A,\hat{R}_E}(\hat{r}_A,\hat{r}_E)=&\frac{4\hat{r}_A\hat{r}_E}{p^2(1-|\rho|^2)+p(\sigma_A^2+\sigma_E^2)+\sigma_A^2\sigma_E^2}
	I_0\left(\frac{2p|\rho|\hat{r}_A\hat{r}_E}{p^2(1-|\rho|^2)+p(\sigma_A^2+\sigma_E^2)+\sigma_A^2\sigma_E^2}\right)\\
	&\exp\left({-\frac{\hat{r}_A^2\left(p+\sigma_E^2\right)+\hat{r}_E^2\left(p+\sigma_A^2\right)}{p^2(1-|\rho|^2)+p(\sigma_A^2+\sigma_E^2)+\sigma_A^2\sigma_E^2}}\right).
	\end{align*}	
	The joint PDF $f_{\hat{R}_B,\hat{R}_E}(\hat{r}_B,\hat{r}_E)$ is similarly obtained, replacing subscripts $A$ by $B$.
\end{lemma}
\begin{proof}
		The proof is given in Appendix~\ref{subsection:proof_RA_RE}.
\end{proof}
As for $h(\hat{R}_A,\hat{R}_B)$, it is difficult to find a closed-form expression of $h(\hat{R}_A,\hat{R}_E)$ and $h(\hat{R}_B,\hat{R}_E)$ due to the presence of the Bessel function. However, they can be evaluated numerically using the PDFs obtained in Lemma~\ref{lemma:joint_dist_RA_RE} so that $I(\hat{R}_A;\hat{R}_E)$ and $I(\hat{R}_B;\hat{R}_E)$ can be evaluated. Still, in specific regimes, closed-form solutions can be found.

In the low correlation regime, when $|\rho|\rightarrow 0$, it is easy to see that $f_{\hat{R}_A,\hat{R}_E}(\hat{r}_A,\hat{r}_E)$ converges to the product of two independent Rayleigh PDFs $f_{\hat{R}_A}(\hat{r}_A)f_{\hat{R}_E}(\hat{r}_E)$ and thus $h(\hat{R}_A,\hat{R}_E)=h(\hat{R}_A)+h(\hat{R}_E)$. As could be expected, we find that $I(\hat{R}_A;\hat{R}_E)=I(\hat{R}_B;\hat{R}_E)=0$ and the secret-key capacity is given by Theorem~\ref{theorem:I_RA_RB_high_SNR}.

In the high SNR and correlation regime, the following lemma shows the limiting behavior of the PDFs of $(\hat{R}_A,\hat{R}_E)$ and $(\hat{R}_B,\hat{R}_E)$, which can be used to obtain a simple closed-form expression of $I(\hat{R}_A;\hat{R}_E)$ and $I(\hat{R}_B;\hat{R}_E)$.

\begin{lemma} \label{lemma:joint_dist_RA_RE_high_SNR}
	Under $\mathbf{(As2)}$, as $|\rho|\rightarrow 1$, $\sigma_A^2\rightarrow 0$ and $\sigma_E^2\rightarrow 0$, the PDF $f_{\hat{R}_A,\hat{R}_E}(\hat{r}_A,\hat{r}_E)$ asymptotically converges to
	\begin{align*}
	f_{\hat{R}_A,\hat{R}_E}(\hat{r}_A,\hat{r}_E)&= \frac{2\hat{r}_Ee^{ -\frac{\hat{r}_E^2}{p} }}{p}\frac{ e^{ -\frac{(\hat{r}_A-|\rho|\hat{r}_E)^2}{p(1-|\rho|^2)+\sigma_A^2+\sigma_E^2} }}{\sqrt{\pi (p(1-|\rho|^2)+\sigma_A^2+\sigma_E^2)}}+O\left(\sqrt{1-|\rho|^2+\sigma_A^2}\right),
	\end{align*}
	which corresponds to the product of a Rayleigh and a normal distribution. The same results holds for $f_{\hat{R}_B,\hat{R}_E}(\hat{r}_B,\hat{r}_E)$, replacing subscripts $A$ by $B$, under $\mathbf{(As3)}$.
\end{lemma}
\begin{proof}
	The proof is given in Appendix~\ref{subsection:proof_RA_RE}.
\end{proof}

\begin{theorem} \label{theorem:I_RA_RE_high_SNR_correlation}
	Under $\mathbf{(As2)}$, as $|\rho|\rightarrow 1$, $\sigma_A^2\rightarrow 0$ and $\sigma_E^2\rightarrow 0$, the mutual information $I(\hat{R}_A;\hat{R}_E)$ converges to
	\begin{align*}
	I(\hat{R}_A;\hat{R}_E)\rightarrow &\frac{1}{2} \log_2
\left(\frac{p}{p(1-|\rho|^2)+\sigma_A^2+\sigma_E^2}\right)-\chi,
	\end{align*}
	where the constant penalty $\chi$ is defined in Theorem~\ref{theorem:I_RA_RB_high_SNR}. The mutual information $I(\hat{R}_B;\hat{R}_E)$ can be similarly approximated by replacing subscripts $A$ by $B$, under $\mathbf{(As3)}$.
\end{theorem}
\begin{proof}
	The proof is given in Appendix~\ref{subsection:proof_RA_RE}.
\end{proof}

Using the result of Theorem~\ref{theorem:I_RA_RE_high_SNR_correlation}, we can evaluate the lower bound on the secret-key capacity (\ref{eq:lower_bounds_R}) in the high SNR, high correlated regime, which is given in the right column of Table~\ref{table:1}. As compared with the complex case, the only difference is the \textit{pre-log factor} of 1/2 for envelope sampling. Note that the constant penalty $\chi$ has canceled since it is also present in $I(\hat{R}_A;\hat{R}_B)$. As for the complex case, the lower bound is not restricted to be positive, in which case it is useless. The condition (\ref{eq:lower_bound_N_Z}) for having a larger-than-zero lower bound, which was derived in the complex case, also applies here.

\subsection{Conditional Mutual Information between Alice and Bob}

As shown in (\ref{eq:def_joint_diff_entropy3}) in the complex case, to compute the conditional mutual information $I(\hat{R}_A;\hat{R}_B|\hat{R}_E)$, we need to evaluate the joint different entropy $h(\hat{R}_A,\hat{R}_B,\hat{R}_E)$. The following lemma gives the joint PDF of $(\hat{R}_A,\hat{R}_B,\hat{R}_E)$.

\begin{lemma} \label{lemma:joint_dist_RA_RB_RE}
	The joint PDF of $(\hat{R}_A,\hat{R}_B,\hat{R}_E)$ is given by
	\begin{align*}
	f_{\hat{R}_A,\hat{R}_B,\hat{R}_E}(\hat{r}_A,\hat{r}_B,\hat{r}_E)=&\frac{8\hat{r}_A\hat{r}_B\hat{r}_E}{|\mat{C}_{\hat{H}_A\hat{H}_B\hat{H}_E}|} G\left(\frac{2p(p(1-|\rho|^2)+\sigma_E^2)\hat{r}_A\hat{r}_B}{|\mat{C}_{\hat{H}_A\hat{H}_B\hat{H}_E}|},\frac{2|\rho| p\sigma_B^2\hat{r}_A\hat{r}_E}{|\mat{C}_{\hat{H}_A\hat{H}_B\hat{H}_E}|},\frac{2|\rho| p\sigma_A^2\hat{r}_B\hat{r}_E}{|\mat{C}_{\hat{H}_A\hat{H}_B\hat{H}_E}|}\right)\\
	&\exp \left( -\frac{\hat{r}_A^2|\mat{C}_{\hat{H}_B\hat{H}_E}|+\hat{r}_B^2|\mat{C}_{\hat{H}_A\hat{H}_E}|+\hat{r}_E^2|\mat{C}_{\hat{H}_A\hat{H}_B}|}{|\mat{C}_{\hat{H}_A\hat{H}_B\hat{H}_E}|} \right),\nonumber
	\end{align*}
	with the definition of the function $G(\alpha_1,\alpha_2,\alpha_3)$	 
	 \begin{align*}
	 G(\alpha_1,\alpha_2,\alpha_3)
	 =&\frac{1}{(2\pi)^2}\int_{0}^{2\pi}\int_{0}^{2\pi} \exp\left({\alpha_1 \cos(\phi_1)+\alpha_2 \cos(\phi_2)+\alpha_3\cos(\phi_2-\phi_1) }\right)   d\phi_1  d\phi_2. 
	 \end{align*}
\end{lemma}
\begin{proof}
	The proof is given in Appendix~\ref{subsection:proof_RA_RB_RE}.
\end{proof}

Here again, computing an analytical expression of the joint differential entropy of $(\hat{R}_A,\hat{R}_B,\hat{R}_E)$ is intricate. However, it can be evaluated numerically, so that $I(\hat{R}_A;\hat{R}_B|\hat{R}_E)$ and thus (\ref{eq:upper_bounds_R}) can be computed.

\section{Numerical Validation}
\label{section:Numerical_validation}


This section aims at numerically validating the analytical results presented in previous sections. The following figures plot the lower bound (LB) and the upper bound (UB) on $C_{s}^{\mathrm{Cplex}}$ and $C_{s}^{\mathrm{Evlpe}}$. We will show many cases where the bounds become tight, as foreseen by the results of previous section. The mutual information quantities $I(\hat{H}_A,\hat{H}_B)$ and $I(\hat{R}_A,\hat{R}_B)$ are also plotted for comparison, as they correspond to the secret-key capacity in the case of uncorrelated observations at Eve, \textit{i.e.}, $C_{s}^{\mathrm{Cplex}}=I(\hat{H}_A,\hat{H}_B)$ and $C_{s}^{\mathrm{Evlpe}}=I(\hat{R}_A,\hat{R}_B)$ for $\rho=0$. They can also be seen as another UB, looser than $I(\hat{H}_A,\hat{H}_B|\hat{H}_E)$ and $I(\hat{R}_A,\hat{R}_B|\hat{R}_E)$.

\subsection{Impact of SNR}
\begin{figure}[!t]
	\centering
	\resizebox{0.6\textwidth}{!}{%
		\Large
%
%
\begin{tikzpicture}

\begin{axis}[%
width=4.520833in,
height=3.565625in,
at={(0.758333in,0.48125in)},
scale only axis,
xmin=0,
xmax=30,
xlabel={SNR [dB]},
ymin=-2,
ymax=10,
ylabel={Secret-key capacity [bits/sample]},
title={$\text{SNR}=p/\sigma_A^2=p/\sigma_B^2=p/\sigma_E^2$, $\rho=0.9$},
legend style={at={(0.03,0.97)},anchor=north west,legend cell align=left,align=left,draw=white!15!black}
]
\addplot [color=blue,solid,line width=1.5pt]
  table[row sep=crcr]{%
0	0.415037499278844\\
3.33333333333333	0.906358316888832\\
6.66666666666667	1.62998880422155\\
10	2.52654581449583\\
13.3333333333333	3.52705079218792\\
16.6666666666667	4.58259264258718\\
20	5.66537127432466\\
23.3333333333333	6.76118304368848\\
26.6666666666667	7.86313135771873\\
30	8.96794706576646\\
};
\addlegendentry{\normalsize $I(\hat{H}_A,\hat{H}_B)$};

\addplot [color=blue,dotted,line width=1.5pt]
  table[row sep=crcr]{%
0	-1\\
3.33333333333333	0.107309364962454\\
6.66666666666667	1.21461872992491\\
10	2.32192809488736\\
13.3333333333333	3.42923745984982\\
16.6666666666667	4.53654682481227\\
20	5.64385618977472\\
23.3333333333333	6.75116555473718\\
26.6666666666667	7.85847491969963\\
30	8.96578428466209\\
};
\addlegendentry{\normalsize $I(\hat{H}_A,\hat{H}_B)$ high SNR};

\addplot [color=blue,solid,line width=1.5pt,mark size=4.0pt,mark=triangle,mark options={solid,rotate=270}]
  table[row sep=crcr]{%
0	0.0885939232689883\\
3.33333333333333	0.221714508729963\\
6.66666666666667	0.483436735043058\\
10	0.929610672108601\\
13.3333333333333	1.58510211282168\\
16.6666666666667	2.42290994411779\\
20	3.38580875549412\\
23.3333333333333	4.42086903203331\\
26.6666666666667	5.49337978273891\\
30	6.58424884749482\\
};
\addlegendentry{\normalsize \normalsize LB on $C_{s}^{\mathrm{Cplex}}$};

\addplot [color=blue,solid,line width=1.5pt,mark size=4.0pt,mark=triangle,mark options={solid,rotate=90}]
  table[row sep=crcr]{%
0	0.216181978153841\\
3.33333333333333	0.396978311606667\\
6.66666666666667	0.664654055693846\\
10	1.07590011435926\\
13.3333333333333	1.68124085751111\\
16.6666666666667	2.47724460443986\\
20	3.41374938226708\\
23.3333333333333	4.43450931404763\\
26.6666666666667	5.49986565175381\\
30	6.58729370655497\\
};
\addlegendentry{\normalsize UB on $C_{s}^{\mathrm{Cplex}}$};

\addplot [color=black!20!red,solid,line width=1.5pt]
  table[row sep=crcr]{%
0	0.199449475543038\\
3.33333333333333	0.189337779325221\\
6.66666666666667	0.393901516307943\\
10	0.734767241559467\\
13.3333333333333	1.16700889964794\\
16.6666666666667	1.65480446687061\\
20	2.17374340234317\\
23.3333333333333	2.70930093459276\\
26.6666666666667	3.25339929862963\\
30	3.81199753153044\\
};
\addlegendentry{\normalsize $I(\hat{R}_A,\hat{R}_B)$};

\addplot [color=black!20!red,dotted,line width=1.5pt]
  table[row sep=crcr]{%
0	-1.18803154201592\\
3.33333333333333	-0.634376859534691\\
6.66666666666667	-0.0807221770534643\\
10	0.472932505427763\\
13.3333333333333	1.02658718790899\\
16.6666666666667	1.58024187039022\\
20	2.13389655287144\\
23.3333333333333	2.68755123535267\\
26.6666666666667	3.2412059178339\\
30	3.79486060031513\\
};
\addlegendentry{\normalsize $I(\hat{R}_A,\hat{R}_B)$ high SNR};

\addplot [color=black!20!red,solid,line width=1.5pt,mark size=4.0pt,mark=triangle,mark options={solid,rotate=270}]
  table[row sep=crcr]{%
0	0.0123619926713712\\
3.33333333333333	0.0544359641474279\\
6.66666666666667	0.15863829530547\\
10	0.357203054293795\\
13.3333333333333	0.6651106264622\\
16.6666666666667	1.06951023672823\\
20	1.54117069772795\\
23.3333333333333	2.0524330085596\\
26.6666666666667	2.58468149640648\\
30	3.13764836548228\\
};
\addlegendentry{\normalsize LB on $C_{s}^{\mathrm{Evlpe}}$};

\addplot [color=black!20!red,solid,line width=1.5pt,mark size=4.0pt,mark=triangle,mark options={solid,rotate=90}]
  table[row sep=crcr]{%
0	-0.0140829873654735\\
3.33333333333333	0.0778899554752983\\
6.66666666666667	0.204911974300432\\
10	0.402722728504849\\
13.3333333333333	0.696624566567671\\
16.6666666666667	1.08617909370072\\
20	1.54755621656236\\
23.3333333333333	2.05171653440364\\
26.6666666666667	2.58006446773402\\
30	3.13764836548228\\
};
\addlegendentry{\normalsize UB on $C_{s}^{\mathrm{Evlpe}}$};

\draw [<->, line width=2] (axis cs:29,6.5) -- (axis cs:29,8.4);
\node[below right, align=left] at (axis cs:24.5,7.7) {\Large $\approx 2.4$};

\draw [<->, line width=2] (axis cs:29,3.0) -- (axis cs:29,3.6);
\node[below right, align=left] at (axis cs:24.75,2.5) {\Large $\approx 0.67$};

\end{axis}
\end{tikzpicture}%
	}
	\caption{Secret-key capacity for complex channel sampling versus envelope sampling as a function of SNR.}
	\label{fig:comparison_capacity_SNR}
	\vspace{-1em}
\end{figure}
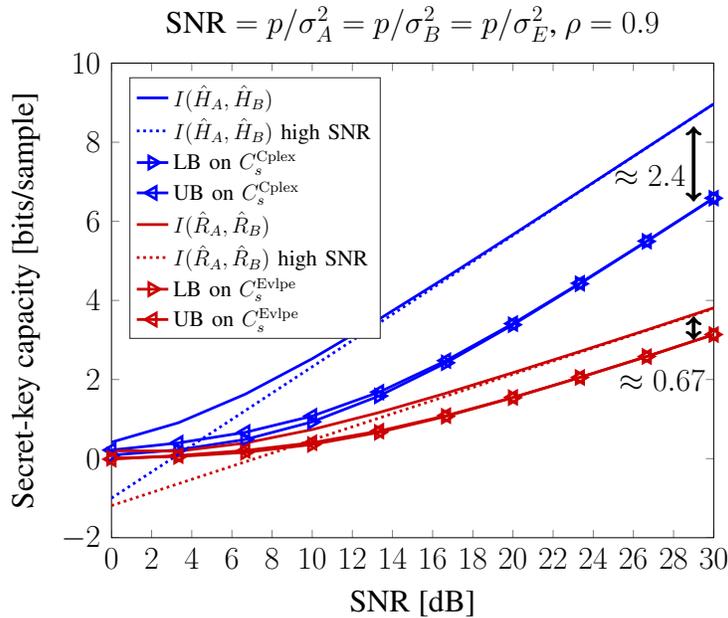

In Fig.~\ref{fig:comparison_capacity_SNR}, the impact of the SNR on $C_{s}^{\mathrm{Cplex}}$ and $C_{s}^{\mathrm{Evlpe}}$ is studied. The SNR is defined as $\text{SNR}=p/\sigma_A^2=p/\sigma_B^2=p/\sigma_E^2$. A first observation is the large performance gain of complex sampling versus envelope sampling. 

Focusing first on the uncorrelated case ($I(\hat{H}_A,\hat{H}_B)$ and $I(\hat{R}_A,\hat{R}_B)$), two penalties of envelope sampling in the high SNR regimen were identified in Table~\ref{table:1}: i) a \textit{pre-log factor} of $1/2$ inducing a smaller slope as a function of SNR and ii) a constant penalty of $\chi$ bit, inducing a translation of the curve downwards of about $0.69$ bit.

In the correlated case ($\rho=0.9$), $C_{s}^{\mathrm{Cplex}}$ and $C_{s}^{\mathrm{Evlpe}}$ are reduced given the knowledge Eve has gained from her channel observations. As foreseen by Prop.~\ref{proposition:tight_high_SNR}, the bounds on $C_{s}^{\mathrm{Cplex}}$ become tight as the SNR grows large and a constant penalty of $\log_2(1-|\rho|^2)\approx -2.4$ bits is observed as compared to the uncorrelated case. Interestingly, the bounds become tight for $C_{s}^{\mathrm{Evlpe}}$, even for smaller values of SNR. The gap as compared to the uncorrelated case can be approximated from Table~\ref{table:1} as $\frac{1}{2}\log_2(1-|\rho|^2)+\chi \approx -0.51$ bits. The inaccuracy with the simulated gap of $-0.67$ bit comes from the fact that the LB on $C_{s}^{\mathrm{Evlpe}}$ in Table~\ref{table:1} only asymptotically holds for $|\rho|\rightarrow 1$.

\subsection{Impact of Correlation}

\begin{figure}[!t]
	\centering
	\resizebox{0.6\textwidth}{!}{%
		\Large
		\input{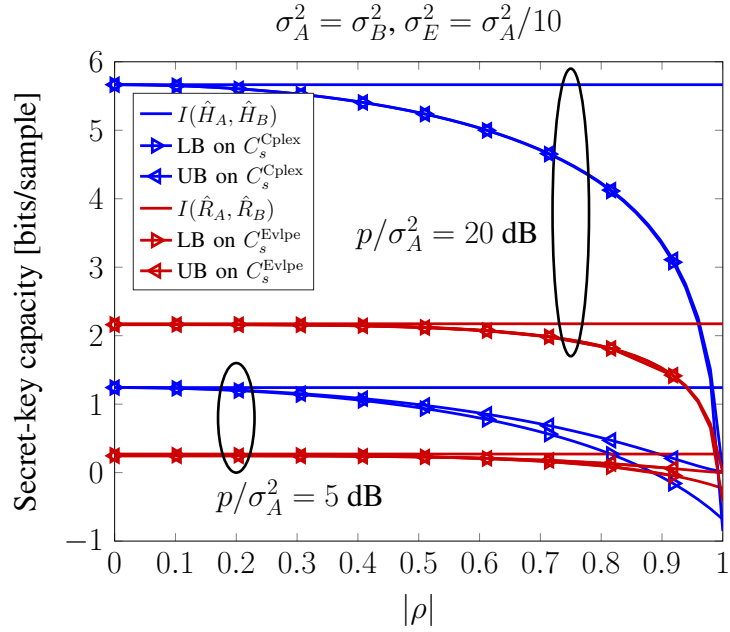}
	}
	\caption{Secret-key capacity for complex channel sampling versus envelope sampling as a function of correlation coefficient magnitude $|\rho|$.}
	\label{fig:comparison_capacity_rho}
	\vspace{-1em}
\end{figure}

In Fig.~\ref{fig:comparison_capacity_rho}, the impact of the correlation coefficient magnitude $|\rho|$ is studied\footnote{From previous analytical studies, it was shown that $C_{s}^{\mathrm{Cplex}}$ and $C_{s}^{\mathrm{Evlpe}}$ only depend on the magnitude of the correlation coefficient and not on its phase.}, for two SNR regimes. We here consider an identical noise variance at Alice and Bob, while Eve uses a more powerful receiver so that $\sigma_A^2=\sigma_B^2$ and $\sigma_E^2=\sigma_A^2/10$.

One can see that, as $|\rho|\rightarrow 0$, the LB and UB become tight and converge to the mutual information between Alice and Bob observations. For larger values of $|\rho|$, bounds are less tight, especially in the complex case. As foreseen by Prop.~\ref{proposition:tight_high_SNR}, for a same value of $|\rho|<1$, the LB and UB become tight for large SNR values. As already discussed in the context of equation (\ref{eq:lower_bound_N_Z}), the LBs on the secret-key capacity are not restricted to be positive. This case is observed in Fig.~\ref{fig:comparison_capacity_rho} for large values of $|\rho|$. Note that this case arises here given the reduced noise power at Eve $\sigma_E^2=\sigma_A^2/10$. In practice, the secret-key capacity cannot be lower than zero. We chose not to put negative values of the LB to zero, as it provides some physical insights on the problem.

\subsection{Impact of Different Noise Variances at Alice and Bob}

\begin{figure}[!t]
	\centering
	\resizebox{0.6\textwidth}{!}{%
		\Large
%
%
\begin{tikzpicture}

\begin{axis}[%
width=4.520833in,
height=3.496181in,
at={(0.758333in,0.550694in)},
scale only axis,
xmin=0,
xmax=30,
xlabel={$p/\sigma_A^2$ [dB]},
ymin=0,
ymax=7,
ylabel={Secret-key capacity [bits/sample]},
title={$\sigma_B^2=\sigma_E^2$, $\rho = 0.6$},
legend style={at={(0.03,0.97)},anchor=north west,legend cell align=left,align=left,draw=white!15!black}
]
\addplot [color=blue,solid,line width=1.5pt]
  table[row sep=crcr]{%
0	0.689365467760937\\
6	1.34820709012628\\
12	1.80922197452356\\
18	1.98789541882185\\
24	2.03939525512655\\
30	2.0528227478971\\
};
\addlegendentry{\normalsize $I(\hat{H}_A,\hat{H}_B)$};

\addplot [color=blue,solid,line width=1.5pt,mark=triangle,mark size=4.0pt,mark options={solid,rotate=270}]
  table[row sep=crcr]{%
0	0.477208370283574\\
6	1.01214817404052\\
12	1.4731630584378\\
18	1.65183650273609\\
24	1.7033363390408\\
30	1.71676383181134\\
};
\addlegendentry{\normalsize LB on $C_{s}^{\mathrm{Cplex}}$};

\addplot [color=blue,solid,line width=1.5pt,mark=triangle,mark size=4.0pt,mark options={solid,rotate=90}]
  table[row sep=crcr]{%
0	0.500577907672661\\
6	1.05206355851738\\
12	1.47815447500029\\
18	1.65223265186438\\
24	1.70336297135418\\
30	1.71676553991671\\
};
\addlegendentry{\normalsize UB on $C_{s}^{\mathrm{Cplex}}$};

\addplot [color=black!20!red,solid,line width=1.5pt]
  table[row sep=crcr]{%
0	0.186081465998807\\
6	0.302819136912714\\
12	0.459150378058709\\
18	0.525646117501578\\
24	0.545279543232846\\
30	0.550430426107978\\
};
\addlegendentry{\normalsize $I(\hat{R}_A,\hat{R}_B)$};

\addplot [color=black!20!red,solid,line width=1.5pt,mark=triangle,mark size=4.0pt,mark options={solid,rotate=270}]
  table[row sep=crcr]{%
0	0.133371239922934\\
6	0.250108910836841\\
12	0.406440151982836\\
18	0.472935891425705\\
24	0.492569317156973\\
30	0.497720200032105\\
};
\addlegendentry{\normalsize LB on $C_{s}^{\mathrm{Evlpe}}$};

\addplot [color=black!20!red,solid,line width=1.5pt,mark=triangle,mark size=4.0pt,mark options={solid,rotate=90}]
  table[row sep=crcr]{%
0	0.0871415069205961\\
6	0.249708822949706\\
12	0.40484706929191\\
18	0.471063583396845\\
24	0.490718545510293\\
30	0.495883810094909\\
};
\addlegendentry{\normalsize UB on $C_{s}^{\mathrm{Evlpe}}$};

\addplot [color=blue,solid,line width=1.5pt,forget plot]
  table[row sep=crcr]{%
0	0.985786140780299\\
6	2.26068224202841\\
12	3.86429598087437\\
18	5.30197399947216\\
24	6.1763664027164\\
30	6.52083898802342\\
};
\addplot [color=blue,solid,line width=1.5pt,mark=triangle,mark size=4.0pt,mark options={solid,rotate=270},forget plot]
  table[row sep=crcr]{%
0	0.702614089375429\\
6	1.77694429961467\\
12	3.27511279905157\\
18	4.6785429915735\\
24	5.54841231994551\\
30	5.89288490525254\\
};
\addplot [color=blue,solid,line width=1.5pt,mark=triangle,mark size=4.0pt,mark options={solid,rotate=90},forget plot]
  table[row sep=crcr]{%
0	0.702691229217582\\
6	1.77724270126324\\
12	3.27618038595605\\
18	4.68157147895157\\
24	5.54931025561312\\
30	5.89295736769487\\
};
\addplot [color=black!20!red,solid,line width=1.5pt,forget plot]
  table[row sep=crcr]{%
0	0.253255980267058\\
6	0.629548768577237\\
12	1.3202698362144\\
18	1.99811633361507\\
24	2.42255093996874\\
30	2.59124014411846\\
};
\addplot [color=black!20!red,solid,line width=1.5pt,mark=triangle,mark size=4.0pt,mark options={solid,rotate=270},forget plot]
  table[row sep=crcr]{%
0	0.162428329750571\\
6	0.562784489026788\\
12	1.2371299527085\\
18	1.90821742769852\\
24	2.33172328945225\\
30	2.50041249360197\\
};
\addplot [color=black!20!red,solid,line width=1.5pt,mark=triangle,mark size=4.0pt,mark options={solid,rotate=90},forget plot]
  table[row sep=crcr]{%
0	0.156697605939694\\
6	0.559942244324855\\
12	1.23352722442023\\
18	1.90460939405087\\
24	2.32719874415703\\
30	2.49505227432741\\
};

\node[below right, align=left, draw=white]
at (axis cs:20.5,1.5) {\Large $\frac{p}{\sigma_B^2}=5$ dB};
\draw [black,line width=0.5mm] (axis cs:28,1.2) ellipse [x radius=3, y radius=100];

\node[below right, align=left, draw=white]
at (axis cs:17.5,4.35) {\Large $\frac{p}{\sigma_B^2}=20$ dB};
\draw [black,line width=0.5mm] (axis cs:26,4.35) ellipse [x radius=3, y radius=220];

\end{axis}
\end{tikzpicture}%
	}
	\caption{Impact of a different noise variance at Alice and Bob.}
	\label{fig:secret_key_capacity_SNR_Alice}
	\vspace{-1em}
\end{figure}
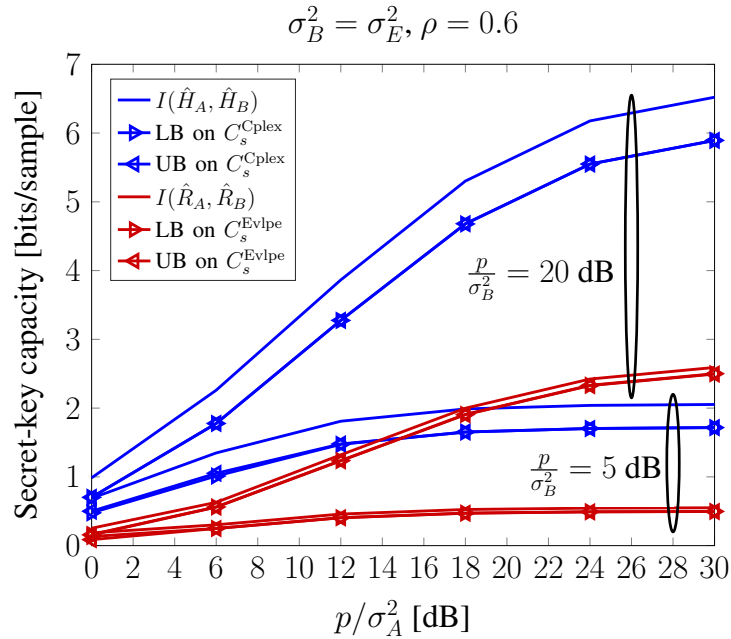

In Fig.~\ref{fig:secret_key_capacity_SNR_Alice}, the impact of a different noise variance at Alice and Bob is studied. More specifically, the SNRs at Bob and Eve are kept identical, \textit{i.e.}, $p/\sigma_B^2=p/\sigma_E^2$, for two SNR regimes (5 dB and 20 dB). On the other hand, the SNR at Alice $p/\sigma_A^2$ is varied from 0 to 30 dB. The correlation coefficient is set to $\rho=0.6$.

As foreseen in Sections~\ref{section:secret_key_capacity_complex} and \ref{section:secret_key_capacity_envelope}, the LB and UB bounds become tight as $\sigma_A^2\rightarrow 0$ for a fixed value of $\sigma_B^2$. Moreover, as $p/\sigma_A^2$ grows large, $C_{s}^{\mathrm{Cplex}}$ and $C_{s}^{\mathrm{Evlpe}}$ saturate at a plateau. This can be explained by the fact that they enter a regime limited by the fixed noise variance at Bob $\sigma_B^2$.

\subsection{Impact of Different Noise Variance at Eve}

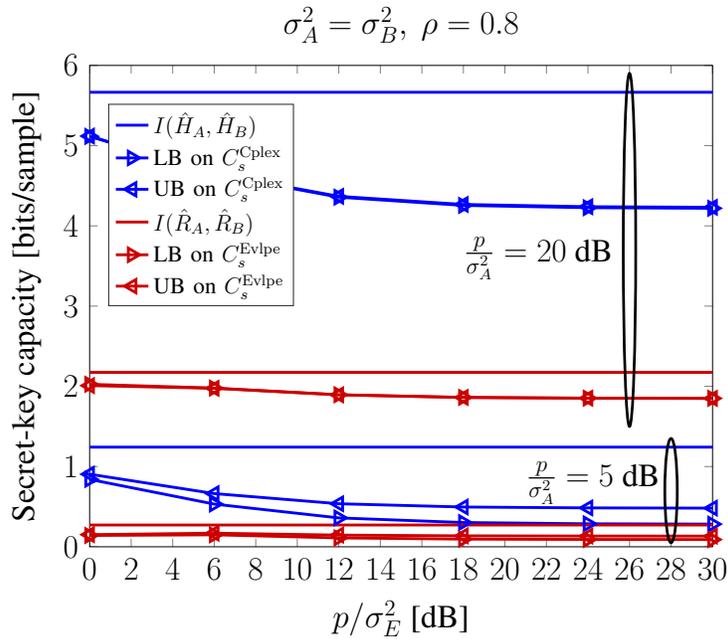
\begin{figure}[!t]
	\centering
	\resizebox{0.6\textwidth}{!}{%
		\Large
%
%
\begin{tikzpicture}

\begin{axis}[%
width=4.520833in,
height=3.496181in,
at={(0.758333in,0.550694in)},
scale only axis,
xmin=0,
xmax=30,
xlabel={$p/\sigma_E^2$ [dB]},
ymin=0,
ymax=6,
ylabel={Secret-key capacity [bits/sample]},
title={$\sigma_A^2=\sigma_B^2,\ \rho=0.8$},
legend style={at={(0.03,0.92)},anchor=north west,legend cell align=left,align=left,draw=white!15!black}
]
\addplot [color=blue,solid,line width=1.5pt]
  table[row sep=crcr]{%
0	1.24200524220002\\
6	1.24200524220002\\
12	1.24200524220002\\
18	1.24200524220002\\
24	1.24200524220002\\
30	1.24200524220002\\
};
\addlegendentry{\normalsize $I(\hat{H}_A,\hat{H}_B)$};

\addplot [color=blue,solid,line width=1.5pt,mark=triangle,mark size=4.0pt,mark options={solid,rotate=270}]
  table[row sep=crcr]{%
0	0.840143607295123\\
6	0.532144426189527\\
12	0.360021106645643\\
18	0.302324085799498\\
24	0.28658132796321\\
30	0.282540637485293\\
};
\addlegendentry{\normalsize LB on $C_{s}^{\mathrm{Cplex}}$};

\addplot [color=blue,solid,line width=1.5pt,mark=triangle,mark size=4.0pt,mark options={solid,rotate=90}]
  table[row sep=crcr]{%
0	0.904840904039149\\
6	0.663114680995913\\
12	0.536401782653026\\
18	0.495539887526251\\
24	0.484544154442089\\
30	0.4817328030039\\
};
\addlegendentry{\normalsize UB on $C_{s}^{\mathrm{Cplex}}$};

\addplot [color=black!20!red,solid,line width=1.5pt]
  table[row sep=crcr]{%
0	0.271328798383817\\
6	0.271328798383817\\
12	0.271328798383817\\
18	0.271328798383817\\
24	0.271328798383817\\
30	0.271328798383817\\
};
\addlegendentry{\normalsize $I(\hat{R}_A,\hat{R}_B)$};

\addplot [color=black!20!red,solid,line width=1.5pt,mark=triangle,mark size=4.0pt,mark options={solid,rotate=270}]
  table[row sep=crcr]{%
0	0.141012770968452\\
6	0.149089071249937\\
12	0.110506046263416\\
18	0.0955260323431142\\
24	0.0913098460768087\\
30	0.0902194351240482\\
};
\addlegendentry{\normalsize LB on $C_{s}^{\mathrm{Evlpe}}$};

\addplot [color=black!20!red,solid,line width=1.5pt,mark=triangle,mark size=4.0pt,mark options={solid,rotate=90}]
  table[row sep=crcr]{%
0	0.151282672950074\\
6	0.169558560328479\\
12	0.145624874167599\\
18	0.136409030833023\\
24	0.133839635671161\\
30	0.133177077468931\\
};
\addlegendentry{\normalsize UB on $C_{s}^{\mathrm{Evlpe}}$};

\addplot [color=blue,solid,line width=1.5pt,forget plot]
  table[row sep=crcr]{%
0	5.66537127432466\\
6	5.66537127432466\\
12	5.66537127432466\\
18	5.66537127432466\\
24	5.66537127432466\\
30	5.66537127432466\\
};
\addplot [color=blue,solid,line width=1.5pt,mark=triangle,mark size=4.0pt,mark options={solid,rotate=270},forget plot]
  table[row sep=crcr]{%
0	5.11568424835103\\
6	4.64664204767288\\
12	4.35760259564334\\
18	4.25503059776644\\
24	4.22647463764185\\
30	4.21910398526711\\
};
\addplot [color=blue,solid,line width=1.5pt,mark=triangle,mark size=4.0pt,mark options={solid,rotate=90},forget plot]
  table[row sep=crcr]{%
0	5.11901683059581\\
6	4.65402607466342\\
12	4.36823279299854\\
18	4.266980385253\\
24	4.23880903097253\\
30	4.2315389093517\\
};
\addplot [color=black!20!red,solid,line width=1.5pt,forget plot]
  table[row sep=crcr]{%
0	2.17374340234317\\
6	2.17374340234317\\
12	2.17374340234317\\
18	2.17374340234317\\
24	2.17374340234317\\
30	2.17374340234317\\
};
\addplot [color=black!20!red,solid,line width=1.5pt,mark=triangle,mark size=4.0pt,mark options={solid,rotate=270},forget plot]
  table[row sep=crcr]{%
0	2.02526401830527\\
6	1.978455334312\\
12	1.8950027289145\\
18	1.86200813582397\\
24	1.85258263720966\\
30	1.8501339453882\\
};
\addplot [color=black!20!red,solid,line width=1.5pt,mark=triangle,mark size=4.0pt,mark options={solid,rotate=90},forget plot]
  table[row sep=crcr]{%
0	2.00892753043911\\
6	1.97553329619109\\
12	1.89382403693293\\
18	1.86140205501943\\
24	1.8521399241662\\
30	1.84973378205917\\
};

\node[below right, align=left]
at (axis cs:20.5,1.2) {\Large $\frac{p}{\sigma_A^2}=5$ dB};
\draw [black,line width=0.5mm] (axis cs:28,0.7) ellipse [x radius=3, y radius=65];

\node[below right, align=left]
at (axis cs:17.5,4) {\Large $\frac{p}{\sigma_A^2}=20$ dB};
\draw [black,line width=0.5mm] (axis cs:26,3.7) ellipse [x radius=3, y radius=220];

\end{axis}
\end{tikzpicture}%
	}
	\caption{Impact of a different noise variance at Eve.}
	\label{fig:secret_key_capacity_SNR_Eve}
	\vspace{-1em}
\end{figure}

In Fig.~\ref{fig:secret_key_capacity_SNR_Eve}, the impact of a different noise variance at Eve is studied. More specifically, the SNRs at Alice and Bob are kept identical, \textit{i.e.}, $p/\sigma_A^2=p/\sigma_B^2$, for two SNR regimes (5 dB and 20 dB). On the other hand, the SNR at Eve $p/\sigma_E^2$ is varied from 0 to 30 dB. The correlation coefficient is set to $\rho=0.8$.

According to Prop.~\ref{proposition:tight_high_SNR}, the LB and UB are tighter in the higher SNR regime. Moreover, as $p/\sigma_E^2$ grows large, $C_{s}^{\mathrm{Cplex}}$ and $C_{s}^{\mathrm{Evlpe}}$ decrease up to a certain floor. This can be explained by the fact that Eve performance is not limited by $\sigma_E^2$ but by the fixed value of the correlation coefficient $\rho$.

\section{Conclusions}
\label{section:conclusion}

In this paper, we have compared the secret-key capacity based on the sampling process of the entire CSI or only its envelope or RSS, taking into account correlation of Eve's observations. We have evaluated lower and upper bounds on the secret-key capacity. In the complex case, we obtain simple closed-form expressions. In the envelope case, the bounds must be evaluated numerically. In a number of particular cases, the lower and upper bounds become tight: low correlation of the eavesdropper, relatively smaller noise variance at Bob than Alice (or vice versa) and specific high SNR regimes. Finally, we have shown that, in the high SNR regime, the bounds can be evaluated in closed-form and result in simple expressions, which highlight the gain of CSI-based systems. The penalty of envelope-based versus complex-based secret-key generation is: i) a \textit{pre-log} factor of $1/2$ instead of $1$, implying a slower slope of the secret-key capacity as a function of SNR and ii) a constant penalty of about $0.69$ bit, which disappears as Eve's channel gets highly correlated.

\section{Appendix}
\label{section:appendix}

\subsection{Upper Bound of Complex Sampling-based Secret-Key Capacity} \label{subsection:proof1}
We need to show that $I(\hat{H}_A;\hat{H}_B|\hat{H}_E)\leq I(\hat{H}_A;\hat{H}_B)$, which is equivalent to showing that
\begin{align*}
0 &\geq I(\hat{H}_A;\hat{H}_B|\hat{H}_E)- I(\hat{H}_A;\hat{H}_B),
\end{align*}
or
\begin{align*}
1 &\geq \frac{|\mat{C}_{\hat{H}_A\hat{H}_E}||\mat{C}_{\hat{H}_B\hat{H}_E}||\mat{C}_{\hat{H}_A\hat{H}_B}|}{(p+\sigma^2_A)(p+\sigma^2_B)(p+\sigma^2_E)|\mat{C}_{\hat{H}_A\hat{H}_B\hat{H}_E}|}\\
0 &\geq \frac{|\mat{C}_{\hat{H}_A\hat{H}_E}||\mat{C}_{\hat{H}_B\hat{H}_E}||\mat{C}_{\hat{H}_A\hat{H}_B}|}{(p+\sigma^2_A)(p+\sigma^2_B)(p+\sigma^2_E)}-|\mat{C}_{\hat{H}_A\hat{H}_B\hat{H}_E}|.
\end{align*}
After computing the expression of each determinant and several simplifications, we obtain
\begin{align*}
\frac{|\mat{C}_{\hat{H}_A\hat{H}_E}||\mat{C}_{\hat{H}_B\hat{H}_E}||\mat{C}_{\hat{H}_A\hat{H}_B}|}{(p+\sigma^2_A)(p+\sigma^2_B)(p+\sigma^2_E)}-|\mat{C}_{\hat{H}_A\hat{H}_B\hat{H}_E}|
=&-|\rho|^2 2p^3+\frac{|\rho p|^4}{p+\sigma^2_E}+|\rho|^2p^4 \left(\frac{1}{p+\sigma^2_A}+\frac{1}{p+\sigma^2_B}\right)\\
&-\frac{|\rho|^4 p^6}{(p+\sigma^2_A)(p+\sigma^2_B)(p+\sigma^2_E)}.
\end{align*}
We still need to prove that this quantity is smaller or equal to zero. We can first simplify the inequality by dividing by $|\rho|^2p^3$. We then need to show that
\begin{align*}
0\geq& -2+ \frac{1}{1+\sigma^2_A/p}+\frac{1}{1+\sigma^2_B/p}+|\rho|^2\frac{1}{1+\sigma^2_E/p}\left(1-\frac{1}{(1+\sigma^2_A/p)(1+\sigma^2_B/p)}\right).
\end{align*}
It is easy to see that the term on the right is maximized for $\sigma^2_E=0$ and $|\rho|=1$ ($|\rho|\leq 1$ by definition). It is then sufficient to focus on that critical case and in particular to show that
\begin{align*}
1&\geq \frac{1}{1+\sigma^2_A/p}+\frac{1}{1+\sigma^2_B/p}-\frac{1}{(1+\sigma^2_A/p)(1+\sigma^2_B/p)}=\frac{1+\sigma^2_A/p+\sigma^2_B/p}{1+\sigma^2_A/p+\sigma^2_B/p+\sigma^2_A\sigma^2_B/p^2},
\end{align*}
which is always smaller or equal to one given that $\sigma^2_A$, $\sigma^2_B$ and $\sigma^2_E$ and $p$ are positive by definition.

\subsection{PDF and Mutual Information of Alice and Eve's Envelopes}
\label{subsection:proof_RA_RE}

In this section, we address the proofs of the results obtained in Propositions~\ref{theorem:liu_Eve} and \ref{proposition:No_loss_info_envelope_Eve}, Lemmas~\ref{lemma:joint_dist_RA_RE} and \ref{lemma:joint_dist_RA_RE_high_SNR} and Theorem~\ref{theorem:I_RA_RE_high_SNR_correlation}. We conduct the proof considering Alice case. The proof can be straightforwardly extended to Bob's case by replacing subscript $A$ by $B$ in all of the following expressions. A starting point is to write the PDF of the channel observations at Alice and Eve. We know that $\hat{H}_A$ and $\hat{H}_E$ follow a ZMCSG with covariance matrix $\mat{C}_{\hat{H}_A\hat{H}_E}$, which gives
\begin{align*}
f_{\hat{H}_A,\hat{H}_E}(\hat{h}_A,\hat{h}_E)&=\frac{1}{\pi^2|\mat{C}_{\hat{H}_A\hat{H}_E}|}e^{-\frac{|\hat{h}_A|^2(p+\sigma_E^2)+|\hat{h}_E|^2(p+\sigma_A^2)-2p\Re(\rho^*\hat{h}_A\hat{h}_E^*)}{|\mat{C}_{\hat{H}_A\hat{H}_E}|}}\nonumber.
\end{align*}
We can express this PDF in polar coordinates using the change of variables $\hat{H}_A=\hat{R}_A \exp(\jmath\hat{\Phi}_A)$, $\hat{H}_E=\hat{R}_E \exp(\jmath\hat{\Phi}_E)$. Doing this, we obtain the joint PDF
\begin{align}
f_{\hat{R}_A,\hat{\Phi}_A,\hat{R}_E,\hat{\Phi}_E}(\hat{r}_A,\hat{\phi}_A,\hat{r}_E,\hat{\phi}_E)\label{distr:HA_H_E_polar_coordinates}
&=\frac{\hat{r}_A\hat{r}_E}{\pi^2|\mat{C}_{\hat{H}_A\hat{H}_E}|}e^{-\frac{\hat{r}_A^2(p+\sigma_E^2)+\hat{r}_E^2(p+\sigma_A^2)-2p\hat{r}_A\hat{r}_E|\rho|\cos(\hat{\phi}_A-\hat{\phi}_E-\angle \rho)}{|\mat{C}_{\hat{H}_A\hat{H}_E}|}}.
\end{align}
We now prove each of the results, relying on (\ref{distr:HA_H_E_polar_coordinates}).

\subsubsection{Complements to the proofs of Propositions~\ref{theorem:liu_Eve} and \ref{proposition:No_loss_info_envelope_Eve}} 

This section derives a set of results on the dependence of random variables, required in the proofs of Propositions~\ref{theorem:liu_Eve} and \ref{proposition:No_loss_info_envelope_Eve}.

Firstly, the random vector $(\hat{\Phi}_A,\hat{\Phi}_E)$ is not independent from $(\hat{R}_A,\hat{R}_E)$, if $|\rho|>1$. Indeed, by simple inspection of (\ref{distr:HA_H_E_polar_coordinates}), we can see that
\begin{align*}
	f_{\hat{R}_A,\hat{\Phi}_A,\hat{R}_E,\hat{\Phi}_E}(\hat{r}_A,\hat{\phi}_A,\hat{r}_E,\hat{\phi}_E)\neq f_{\hat{R}_A,\hat{R}_E}(\hat{r}_A,\hat{r}_E)f_{\hat{\Phi}_A,\hat{\Phi}_E}(\hat{\phi}_A,\hat{\phi}_E).
\end{align*}
The same result holds for $(\hat{\Phi}_A,\hat{\Phi}_B)$ and $(\hat{R}_A,\hat{R}_B)$, as a particularization to the case $\rho=1$ and replacing subscripts $E$ by $B$.

Secondly, $\hat{\Phi}_E$ and $(\hat{R}_A,\hat{R}_E)$ are independent. This can be shown by integrating (\ref{distr:HA_H_E_polar_coordinates}) over $\hat{\phi}_A$ giving
\begin{align}
f_{\hat{R}_A,\hat{R}_E,\hat{\Phi}_E}(\hat{r}_A,\hat{r}_E,\hat{\phi}_E) &=\int_0^{2\pi} f_{\hat{R}_A,\hat{\Phi}_A,\hat{R}_E,\hat{\Phi}_E}(\hat{r}_A,\hat{\phi}_A,\hat{r}_E,\hat{\phi}_E) d\hat{\phi}_A\nonumber\\
&=\frac{2\hat{r}_A\hat{r}_E}{\pi|\mat{C}_{\hat{H}_A\hat{H}_E}|}I_0\left({\frac{2p|\rho|\hat{r}_A\hat{r}_E}{|\mat{C}_{\hat{H}_A\hat{H}_E}|}}\right)e^{-\frac{\hat{r}_A^2(p+\sigma_E^2)+\hat{r}_E^2(p+\sigma_A^2)}{|\mat{C}_{\hat{H}_A\hat{H}_E}|}}\label{eq:f_RA_RE_PhiE},
\end{align}
where $I_0(.)$ is the zero order modified Bessel function of the first kind. This shows that
\begin{align*}
	f_{\hat{R}_A,\hat{R}_E,\hat{\Phi}_E}(\hat{r}_A,\hat{r}_E,\hat{\phi}_E)=f_{\hat{R}_A,\hat{R}_E}(\hat{r}_A,\hat{r}_E)f_{\hat{\Phi}_E}(\hat{\phi}_E).
\end{align*}
The same result holds for $\hat{\Phi}_B$ and $(\hat{R}_A,\hat{R}_B)$, as a particularization to the case $\rho=1$ and replacing subscripts $E$ by $B$.

Thirdly, the envelope and the phase of a ZMCSG are independent. Take for instance the PDF of $\hat{H}_E$, which can be written in polar coordinates, using a change of variable $\hat{H}_E=\hat{R}_E\exp(\jmath\hat{\Phi}_E)$, as
\begin{align*}
f_{\hat{R}_E,\hat{\Phi}_E}(\hat{r}_E,\hat{\phi}_E)&=\frac{\hat{r}_E}{\pi(p+\sigma_E^2)}e^{-\frac{\hat{r}_E^2}{p+\sigma_E^2}},
\end{align*}
which shows that $f_{\hat{R}_E,\hat{\Phi}_E}(\hat{r}_E,\hat{\phi}_E)=f_{\hat{R}_E}(\hat{r}_E)f_{\hat{\Phi}_E}(\hat{\phi}_E)$, implying independence. The same result holds for $\hat{H}_A$ and $\hat{H}_B$.

\subsubsection{Proof of Lemma~\ref{lemma:joint_dist_RA_RE}} \label{proof:lemma:joint_dist_RA_RE}
The joint PDF $f_{\hat{R}_A,\hat{R}_E}(\hat{r}_A,\hat{r}_E)$ can be obtained by integrating (\ref{eq:f_RA_RE_PhiE}) over $\hat{\phi}_E$, which gives
\begin{align}
f_{\hat{R}_A,\hat{R}_E}(\hat{r}_A,\hat{r}_E) &=\int_0^{2\pi} f_{\hat{R}_A,\hat{R}_E,\hat{\Phi}_E}(\hat{r}_A,\hat{r}_E,\hat{\phi}_E) d\hat{\phi}_E\nonumber\\
&=\frac{4\hat{r}_A\hat{r}_E}{|\mat{C}_{\hat{H}_A\hat{H}_E}|}I_0\left({\frac{2p|\rho|\hat{r}_A\hat{r}_E}{|\mat{C}_{\hat{H}_A\hat{H}_E}|}}\right)e^{-\frac{\hat{r}_A^2(p+\sigma_E^2)+\hat{r}_E^2(p+\sigma_A^2)}{|\mat{C}_{\hat{H}_A\hat{H}_E}|}}\label{eq:f_RA_RE},
\end{align}
and leads to the result of Lemma~\ref{lemma:joint_dist_RA_RE}, noting that $|\mat{C}_{\hat{H}_A\hat{H}_E}|=p^2(1-|\rho|^2)+p(\sigma_A^2+\sigma_E^2)+\sigma_A^2\sigma_E^2$.

\subsubsection{Proof of Lemma~\ref{lemma:joint_dist_RA_RE_high_SNR}} \label{proof:joint_dist_RA_RE_high_SNR}

From Bessel function theory \cite[Eq.~10.40.1]{NIST:DLMF}, we know that, as $r\rightarrow +\infty$,
\begin{align}
	I_0(r)&=\frac{e^r}{\sqrt{2\pi r}}+\epsilon_0,\ |\epsilon_0|=O\left(\frac{e^r}{r^{3/2}}\right). \label{eq:Bessel_error_order}
\end{align}
In our case, we have
\begin{align}
	r={\frac{2p|\rho|\hat{r}_A\hat{r}_E}{|\mat{C}_{\hat{H}_A\hat{H}_E}|}}=\frac{2p|\rho|\hat{r}_A\hat{r}_E}{(1-|\rho|^2)p^2+p(\sigma_E^2+\sigma_A^2)+\sigma_E^2\sigma_A^2} . \label{eq:def_r}
\end{align}
The Bessel asymptotic expansion is thus accurate when $r$ becomes large. This is precisely the case as $\sigma_A^2\rightarrow 0$, $\sigma_E^2\rightarrow 0$ and $|\rho|\rightarrow 1$, for $\hat{r}_A>0$ and $\hat{r}_E>0$. Using the Bessel asymptotic expansion of $I_0(.)$ in (\ref{eq:f_RA_RE}), we get
\begin{align}
f_{\hat{R}_A,\hat{R}_E}(\hat{r}_A,\hat{r}_E)=&
\frac{2}{p}\sqrt{\frac{\hat{r}_A\hat{r}_E}{ |\rho|}}e^{-\frac{ \hat{r}_A^2\sigma_E^2+\hat{r}_E^2(\sigma_A^2+p(1-|\rho|^2))}{|\mat{C}_{\hat{H}_A\hat{H}_E}|}}\frac{1}{\sqrt{\pi|\mat{C}_{\hat{H}_A\hat{H}_E}|/p}}e^{-\frac{(\hat{r}_A-|\rho|\hat{r}_E)^2}{|\mat{C}_{\hat{H}_A\hat{H}_E}|/p}} +\epsilon_1\label{eq:first_approx},
\end{align}
where $\epsilon_1$ is the approximation error
\begin{align*}
	\epsilon_1&=\frac{4 \hat{r}_A\hat{r}_E}{|\mat{C}_{\hat{H}_A\hat{H}_E}|} \exp \left(-\frac{\hat{r}_A^2(p+\sigma_E^2)+\hat{r}_E^2(p+\sigma_A^2)}{|\mat{C}_{\hat{H}_A\hat{H}_E}|}\right)\epsilon_0.
\end{align*}
Note that, in the particular cases $\hat{r}_A=0$ or $\hat{r}_E=0$, $\epsilon_1=0$ since $(\ref{eq:first_approx})=(\ref{eq:f_RA_RE})=0$. Using $(\ref{eq:Bessel_error_order})$ and the definition of $r$ in (\ref{eq:def_r}), we can bound the error $\epsilon_1$ as follows
\begin{align*}
	|\epsilon_1|&=O\left( \frac{\left({|\mat{C}_{\hat{H}_A\hat{H}_E}|}\right)^{1/2}e^{ -\frac{\hat{r}_A^2\left(p+\sigma_E^2\right)+\hat{r}_E^2\left(p+\sigma_A^2\right)-2p|\rho|\hat{r}_A\hat{r}_E}{|\mat{C}_{\hat{H}_A\hat{H}_E}|} }}{\left(p|\rho|\right)^{3/2}\left({\hat{r}_A\hat{r}_E}\right)^{1/2}} \right)\\
	&=O\left( \sqrt{{1-|\rho|^2+\sigma_A^2+\sigma_E^2}} \right),
\end{align*}
where we used the fact that the exponential can be bounded in the asymptotic regime by an independent constant.
The second term exponential term of $(\ref{eq:first_approx})$ suggests the following approximation $\hat{r}_A\approx|\rho|\hat{r}_E$. We thus obtain
\begin{align}
f_{\hat{R}_A,\hat{R}_E}(\hat{r}_A,\hat{r}_E)=&\frac{2\hat{r}_E{e^{  -\hat{r}_E^2\frac{p(1-|\rho|^2)+|\rho|^2\sigma_E^2+\sigma_A^2}{|\mat{C}_{\hat{H}_A\hat{H}_E}|} }}}{p}\frac{e^{-\frac{(\hat{r}_A-|\rho|\hat{r}_E)^2}{|\mat{C}_{\hat{H}_A\hat{H}_E}|/p}  }}{\sqrt{\pi |\mat{C}_{\hat{H}_A\hat{H}_E}|/p}}+\epsilon_1+\epsilon_2,\label{eq:second_approx}
\end{align}
where $\epsilon_2$ is the approximation error related to this second approximation
\begin{align*}
	\epsilon_2=&\frac{2}{p}\frac{e^{-\frac{(\hat{r}_A-|\rho|\hat{r}_E)^2}{|\mat{C}_{\hat{H}_A\hat{H}_E}|/p}  }}{\sqrt{\pi |\mat{C}_{\hat{H}_A\hat{H}_E}|/p}}\left(\sqrt{\frac{\hat{r}_A\hat{r}_E}{ |\rho|}}e^{-\frac{ \hat{r}_A^2\sigma_E^2+\hat{r}_E^2(\sigma_A^2+p(1-|\rho|^2))}{|\mat{C}_{\hat{H}_A\hat{H}_E}|}} -\hat{r}_Ee^{  -\hat{r}_E^2\frac{p(1-|\rho|^2)+|\rho|^2\sigma_E^2+\sigma_A^2}{|\mat{C}_{\hat{H}_A\hat{H}_E}|} }\right).
\end{align*}
When $\hat{r}_A=|\rho|\hat{r}_E$, the term in parenthesis is exactly zero and so $\epsilon_2=0$. In other cases, it can be bounded by an independent constant as $\sigma_A^2\rightarrow 0$, $\sigma_E^2\rightarrow 0$ and $|\rho|\rightarrow 1$, giving 
\begin{align*}
	|\epsilon_2|=&
	O\left(\frac{e^{-\frac{\beta}{(1-|\rho|^2)+\sigma_A^2+\sigma_E^2}  }}{\sqrt{1-|\rho|^2+\sigma_A^2+\sigma_E^2}}\right),
\end{align*}
where $\beta$ is some real strictly positive constant. Moreover, we can still simplify (\ref{eq:second_approx}) by performing the two following approximations $|\mat{C}_{\hat{H}_A\hat{H}_E}|/p\approx  p(1-|\rho|^2)+\sigma_A^2+\sigma_E^2$ and $\frac{p(1-|\rho|^2)+|\rho|^2\sigma_E^2+\sigma_A^2}{|\mat{C}_{\hat{H}_A\hat{H}_E}|}\approx 1/p$ so that we get
\begin{align*}
f_{\hat{R}_A,\hat{R}_E}(\hat{r}_A,\hat{r}_E)=&\frac{2\hat{r}_Ee^{  -\frac{\hat{r}_E^2}{p} }}{p}\frac{e^{-\frac{(\hat{r}_A-|\rho|\hat{r}_E)^2}{ p(1-|\rho|^2)+\sigma_A^2+\sigma_E^2}  }}{\sqrt{\pi  (p(1-|\rho|^2)+\sigma_A^2+\sigma_E^2)}}+ \epsilon_1+\epsilon_2+\epsilon_3+\epsilon_4,
\end{align*}
which gives the asymptotic distribution of Lemma~\ref{lemma:joint_dist_RA_RE_high_SNR} and 
where $\epsilon_3$ and $\epsilon_4$ are the approximation errors related to the approximations
\begin{align*}
	\epsilon_3&=\frac{2\hat{r}_E}{p\sqrt{\pi |\mat{C}_{\hat{H}_A\hat{H}_E}|/p}}\left( {e^{  -\hat{r}_E^2\frac{p(1-|\rho|^2)+|\rho|^2\sigma_E^2+\sigma_A^2}{|\mat{C}_{\hat{H}_A\hat{H}_E}|} -\frac{p(\hat{r}_A-|\rho|\hat{r}_E)^2}{|\mat{C}_{\hat{H}_A\hat{H}_E}|}  }}  -e^{  -\frac{\hat{r}_E^2}{p} }{e^{-\frac{(\hat{r}_A-|\rho|\hat{r}_E)^2}{ p(1-|\rho|^2)+\sigma_A^2+\sigma_E^2}  }}\right)\\
	\epsilon_4	&=\frac{2\hat{r}_Ee^{  -\frac{\hat{r}_E^2}{p} }{e^{-\frac{(\hat{r}_A-|\rho|\hat{r}_E)^2}{ p(1-|\rho|^2)+\sigma_A^2+\sigma_E^2}  }}}{p\sqrt{\pi}}\left(\frac{1}{\sqrt{|\mat{C}_{\hat{H}_A\hat{H}_E}|/p}} - \frac{1}{\sqrt{p(1-|\rho|^2)+\sigma_A^2+\sigma_E^2}} \right).
\end{align*}	
To bound $\epsilon_3$ and $\epsilon_4$, we can use a first order Taylor expansion of the exponential and the inverse of a square root respectively. We find
\begin{align*}
	|\epsilon_3|
	&=O\left(\frac{{(1-|\rho|^2)\sigma_E^2+\sigma_A^2\sigma_E^2}}{  (1-|\rho|^2+\sigma_A^2+\sigma_E^2)^{3/2}}\right)\\
	|\epsilon_4|&=O\left(\frac{\sigma_A^2+\sigma_E^2}{\sqrt{1-|\rho|^2+\sigma_A^2+\sigma_E^2}}\right). 
\end{align*}
Finally, combining the bounds on the approximation errors $\epsilon_1,\epsilon_2,\epsilon_3$ and $\epsilon_4$, we find that the total approximation error can be bounded as
\begin{align*}
	|\epsilon_1+\epsilon_2+\epsilon_3+\epsilon_4|
	&=O\left(\sqrt{1-|\rho|^2+\sigma_A^2}\right),
\end{align*}
where we used $\mathbf{(As2)}$. This completes the proof.


\subsubsection{Proof of Theorem~\ref{theorem:I_RA_RE_high_SNR_correlation}} \label{proof:I_RA_RE_high_SNR_correlation}

Let us define the asymptotic PDF of $f_{\hat{R}_A,\hat{R}_E}(\hat{r}_A,\hat{r}_E)$ as
\begin{align*}
f^{\mathrm{High}}_{\hat{R}_A,\hat{R}_E}(\hat{r}_A,\hat{r}_E)=&\frac{2\hat{r}_Ee^{  -\frac{\hat{r}_E^2}{p} }}{p}\frac{e^{-\frac{(\hat{r}_A-|\rho|\hat{r}_E)^2}{ p(1-|\rho|^2)+\sigma_A^2+\sigma_E^2}  }}{\sqrt{\pi  (p(1-|\rho|^2)+\sigma_A^2+\sigma_E^2)}}.
\end{align*}
We can see that the PDF factorizes as $f^{\mathrm{High}}_{\hat{R}_A,\hat{R}_E}(\hat{r}_A,\hat{r}_E)={f}_{1}(\hat{r}_E){f}_{2}(\hat{r}_A|\hat{r}_E)$. We can identify ${f}_{1}(\hat{r}_E)$ to be a Rayleigh distribution with parameter $\frac{p}{2}$, while the conditional PDF ${f}_{2}(\hat{r}_A|\hat{r}_E)$ is a normal centered in $|\rho|\hat{r}_E$ and of variance $({p(1-|\rho|^2)+\sigma_A^2+\sigma_E^2})/{2}$. 

Results such as \cite[Th.~1]{Godavarti2004} can be used to prove that, for a sequence of PDFs such that $f^{\mathrm{High}}_{\hat{R}_A,\hat{R}_E}(\hat{r}_A,\hat{r}_E)\rightarrow f_{\hat{R}_A,\hat{R}_E}(\hat{r}_A,\hat{r}_E)$ pointwise, their differential entropy also converges provided that: i) their second order moments are bounded from above and ii) their PDF is bounded from above. These two conditions are satisfied in our case as long as $p$, $\sigma_A^2$ and $\sigma_E^2$ are bounded from above, which makes practical sense. In the pathological case $\sigma_A^2=0$, $\sigma_E^2=0$ or $|\rho|=1$, $|\mat{C}_{\hat{H}_A\hat{H}_E}|=0$ and the PDFs are unbounded, which makes practical sense since $h(\hat{R}_A,\hat{R}_E)\rightarrow-\infty$. Unfortunately, finding the analytical rate of convergence of the differential entropy is intricate.

All of the following expressions should be understood in the asymptotic sense as $\sigma_A^2\rightarrow 0$ and $\sigma_E^2\rightarrow 0$ and $|\rho|\rightarrow 1$. Using the chain rule for the differential entropy $h(X,Y)=h(X)+h(Y|X)$, the general expression of the differential entropies of Rayleigh and normal distributions, the joint differential entropy of the distribution $f^{\mathrm{High}}_{\hat{R}_A,\hat{R}_E}(\hat{r}_A,\hat{r}_E)$ can be easily computed and we find
\begin{align*}
h(\hat{R}_A,\hat{R}_E)\rightarrow&\frac{1}{2}\log_2 \left(p^2(1-|\rho|^2)+p(\sigma_A^2+\sigma_E^2)\right)  	+ \frac{1}{2}\log_2 \left(\frac{\pi e^{3+\gamma}}{4}\right) .
\end{align*}
Inserting this expression in (\ref{eq:def_I_RA_RE}), together with the expressions of $h(\hat{R}_A)$ and $h(\hat{R}_E)$ given in (\ref{eq:h_RA}) and (\ref{eq:h_RE}) respectively, we finally obtain
\begin{align*}
I(\hat{R}_A,\hat{R}_E)\rightarrow
&\frac{1}{2}\log_2 \left(\frac{(p+\sigma_A^2)(p+\sigma_E^2)}{p^2(1-|\rho|^2)+p(\sigma_A^2+\sigma_E^2)}\right) + \chi\\
\rightarrow&\frac{1}{2}\log_2 \left(\frac{p}{p(1-|\rho|^2)+\sigma_A^2+\sigma_E^2}\right) + \chi,
\end{align*}
with the definition of $\chi$ introduced in Theorem~\ref{theorem:I_RA_RB_high_SNR}, which concludes the proof.

\subsection{PDF of Alice, Bob and Eve's Envelopes}
\label{subsection:proof_RA_RB_RE}

We know that $\hat{H}_A$, $\hat{H}_B$ and $\hat{H}_E$ follow a ZMCSG with covariance matrix $\mat{C}_{\hat{H}_A\hat{H}_B\hat{H}_E}$, which gives
\begin{align*}
f_{\hat{H}_A,\hat{H}_B,\hat{H}_E}(\hat{h}_A,\hat{h}_B,\hat{h}_E)=&\frac{1}{\pi^3|\mat{C}_{\hat{H}_A\hat{H}_B\hat{H}_E}|}e^{ \frac{ 2p(p(1-|\rho|^2)+\sigma_E^2)\hat{h}_A\hat{h}_B^*+2 p\sigma_B^2\Re(\hat{h}_A\rho^*\hat{h}_E^*)+2 p\sigma_A^2\Re(\hat{h}_B\rho^*\hat{h}_E^*)}{|\mat{C}_{\hat{H}_A\hat{H}_B\hat{H}_E}|} }\nonumber\\
&e^{-\frac{|\hat{h}_A|^2|\mat{C}_{\hat{H}_B\hat{H}_E}|+|\hat{h}_B|^2|\mat{C}_{\hat{H}_A\hat{H}_E}|+|\hat{h}_E|^2|\mat{C}_{\hat{H}_A\hat{H}_B}|}{|\mat{C}_{\hat{H}_A\hat{H}_B\hat{H}_E}|}}.
\end{align*}
This PDF can be expressed in polar coordinates as
\begin{align}
&f_{\hat{R}_A,\hat{R}_B,\hat{R}_E,\hat{\Phi}_A,\hat{\Phi}_B,\hat{\Phi}_E}(\hat{r}_A,\hat{r}_B,\hat{r}_E,\hat{\phi}_A,\hat{\phi}_B,\hat{\phi}_E)=\frac{\hat{r}_A\hat{r}_B\hat{r}_E}{\pi^3|\mat{C}_{\hat{H}_A\hat{H}_B\hat{H}_E}|}e^{{-\frac{\hat{r}_A^2|\mat{C}_{\hat{H}_B\hat{H}_E}|+\hat{r}_B^2|\mat{C}_{\hat{H}_A\hat{H}_E}|+\hat{r}_E^2|\mat{C}_{\hat{H}_A\hat{H}_B}|}{|\mat{C}_{\hat{H}_A\hat{H}_B\hat{H}_E}|}}}\label{eq:refeq}\\
& e^{{ \frac{ 2p(p(1-|\rho|^2)+\sigma_E^2)\hat{r}_A\hat{r}_B\cos(\hat{\phi}_A-\hat{\phi}_B)+2 p\sigma_B^2\hat{r}_A\hat{r}_E|\rho|\cos(\hat{\phi}_A-\hat{\phi}_E-\angle \rho)+2 p\sigma_A^2\hat{r}_B\hat{r}_E|\rho|\cos(\hat{\phi}_B-\hat{\phi}_E-\angle \rho)}{|\mat{C}_{\hat{H}_A\hat{H}_B\hat{H}_E}|} }}.\nonumber
\end{align}
The joint PDF $f_{\hat{R}_A,\hat{R}_B,\hat{R}_E}(\hat{r}_A,\hat{r}_B,\hat{r}_E)$ can be obtained by integrating (\ref{eq:refeq}) over the phases $\hat{\phi}_A$, $\hat{\phi}_B$ and $\hat{\phi}_E$, which leads to the result of Lemma~\ref{lemma:joint_dist_RA_RB_RE}. Indeed the first two terms do not depend on the phases, so that they can be put out of the integrals. The third term however does. One can easily see that the phase of $\rho$ does not impact the result, so that it can be removed. One can further notice that the cosines do not depend on the absolute phases $\hat{\phi}_A,\hat{\phi}_B,\hat{\phi}_E$ but on their differences. Making a change of variable $\phi_1=\hat{\phi}_A-\hat{\phi}_B$, $\phi_2=\hat{\phi}_A-\hat{\phi}_E$, we see that the last difference is $\hat{\phi}_B-\hat{\phi}_E=\phi_2-\phi_1$. Hence, one integral simplifies.

%
%

\ifCLASSOPTIONcaptionsoff
  \newpage
\fi



\bibliographystyle{IEEEtran}

\bibliography{IEEEabrv,refs}
\end{document}